\documentclass[11pt]{article}
\usepackage{graphicx,amsmath,amsfonts,amssymb,bm,hyperref,url,breakurl,epsfig,epsf,color,fullpage,MnSymbol,mathbbol, 
} 
\usepackage[small,bf]{caption}
\usepackage{hyperref}
\definecolor{darkred}{RGB}{150,0,0}
\definecolor{darkgreen}{RGB}{0,150,0}
\definecolor{darkblue}{RGB}{0,0,200}
\hypersetup{colorlinks=true, linkcolor=darkred, citecolor=darkgreen, urlcolor=darkblue}

\newtheorem{theorem}{Theorem}[section]
\newtheorem{lemma}[theorem]{Lemma}

\newcommand{\opnorm}[1]{\left\|#1\right\|}

\newcommand{\onenorm}[1]{\left\|#1\right\|_{\ell_1}}
\newcommand{\twonorm}[1]{\left\|#1\right\|_{\ell_2}}

\newcommand{\infnorm}[1]{\left\|#1\right\|_{\ell_\infty}}

\newcommand{\abs}[1]{\left|#1\right|}

\newcommand{\R}{\mathbb{R}}

\newcommand{\<}{\langle}
\renewcommand{\>}{\rangle}
\newcommand{\Var}{\textrm{Var}}
\newcommand{\sgn}[1]{\textrm{sgn}(#1)}
\renewcommand{\P}{\operatorname{\mathbb{P}}}
\newcommand{\E}{\operatorname{\mathbb{E}}}

\renewcommand{\i}{\imath}

\newcommand{\supp}[1]{\operatorname{supp}(#1)}

\newcommand{\Id}{\text{\em I}}

\numberwithin{equation}{section}

\def \endprf{\hfill {\vrule height6pt width6pt depth0pt}\medskip}
\newenvironment{proof}{\noindent {\bf Proof} }{\endprf\par}

\title{A Probabilistic and RIPless Theory of Compressed Sensing}
\author{Emmanuel J. Cand\`es$^1$ and Yaniv Plan$^2$\\
  \vspace{-.1cm}\\
  $^1$Departments of Mathematics and of Statistics, Stanford
  University, Stanford,
  CA 94305\\
  \vspace{-.3cm}\\
  $^2$Applied and Computational Mathematics, Caltech, Pasadena, CA
  91125}
  
\date{November 2010}

\begin{document}

\vspace{-0.3in}

\maketitle
\begin{abstract}
  This paper introduces a simple and very general theory of
  compressive sensing. In this theory, the sensing mechanism simply
  selects sensing vectors independently at random from a probability
  distribution $F$; it includes all models --- e.g.~Gaussian,
  frequency measurements --- discussed in the literature, but also
  provides a framework for new measurement strategies as well. We
  prove that if the probability distribution $F$ obeys a simple
  incoherence property and an isotropy property, one can faithfully
  recover approximately sparse signals from a minimal number of noisy
  measurements. The novelty is that our recovery results do not
  require the restricted isometry property (RIP) --- they make use of a
  much weaker notion --- or a random model for the signal.  As an
  example, the paper shows that a signal with $s$ nonzero entries can
  be faithfully recovered from about $s \log n$ Fourier coefficients
  that are contaminated with noise.
%
≈ƒ  
\end{abstract}


\bigskip

\begin{center}
{\em Dedicated to the memory of Jerrold E. Marsden.}
\end{center}

\section{Introduction}
\label{sec:intro}

This paper develops a novel, simple and general theory of compressive
sensing \cite{CRT06,CT06,donoho2006compressed}, a rapidly growing
field of research that has developed protocols for acquiring certain
types of signals with far fewer data bits than what is classically
accepted.

\subsection{A RIPless theory?}

The early paper \cite{CRT06} triggered a massive amount of research by
showing that it is possible to sample signals at a rate proportional
to their information content rather than their bandwidth. For
instance, in a discrete setting, this theory asserts that a digital
signal $x \in \R^n$ (which can be viewed as Nyquist samples of a
continuous-time signal over a time window of interest) can be
recovered from a small random sample of its Fourier coefficients
provided that $x$ is sufficiently sparse. Formally, suppose that our
signal $x$ has at most $s$ nonzero amplitudes at completely unknown
locations --- such a signal is called {\em $s$-sparse} --- and that we are
given the value of its discrete Fourier transform (DFT) at $m$
frequencies selected uniformly at random (we think of $m$ as being
much smaller than $n$). Then \cite{CRT06} showed
that one can recover $x$ by solving an optimization problem
which simply finds, among all candidate signals, that with the
minimum $\ell_1$ norm; the number of samples we need must be on the
order of $s \log n$. In other words, if we think of $s$ as a measure
of the information content, we can sample {\em nonadaptively} nearly
at the information rate without information loss. By swapping time and
frequency, this also says that signals occupying a very large
bandwidth but with a sparse spectrum can be sampled (at random time
locations) at a rate far below the Shannon-Nyquist rate.

Despite considerable progress in the field, some important questions
are still open. We discuss two that have both a theoretical and
practical appeal.  
\begin{quote}
  {\em Is it possible to faithfully recover a nearly sparse signal $x
    \in \R^n$, one which is well approximated by its $s$ largest
    entries, from about $s \log n$ of its Fourier coefficients? Is it
    still possible when these coefficients are further corrupted by
    noise?}
\end{quote}
These issues are paramount since in real-world applications, signals
are never exactly sparse, and measurements are never perfect either. Now the traditional way of addressing these types of problems
in the field is by means of the {\em restricted isometry property}
(RIP) \cite{CT05}. The trouble here is that it is unknown whether or
not this property holds when the sample size $m$ is on the order of $s
\log n$. In fact, answering this one way or the other is generally
regarded as extremely difficult, and so the restricted isometry
machinery does not directly apply in this setting.

This paper proves that the two questions formulated above have
positive answers. In fact, we introduce recovery results which
are --- up to a logarithmic factor --- as good as those one would get if
the restricted isometry property were known to be true. To fix ideas,
suppose we observe $m$ noisy discrete Fourier coefficients about an 
$s$-sparse signal $x$,
\begin{equation}
  \label{eq:fourier}
  \tilde{y}_k = \sum_{t = 0}^{n-1} \, e^{-\i 2\pi \omega_k t} x[t] + \sigma z_k,
  \quad k = 1, \ldots, m.  
\end{equation}
Here, the frequencies $\omega_k$ are chosen uniformly at random in
$\{0, 1/n, 2/n, \ldots, (n-1)/n\}$ and $z_k$ is white noise with unit
variance. Then if the number of samples $m$ is on the order of $s \log
n$, it is possible to get an estimate $\hat x$ obeying
\begin{equation}
\label{eq:polylog}
\|\hat x - x\|_{\ell_2}^2 = \text{polylog}(n) \, \frac{s}{m} \sigma^2 
\end{equation}
by solving a convex $\ell_1$-minimization program. (Note that when the
noise vanishes, the recovery is exact.)  Up to the logarithmic factor,
which may sometimes be on the order of $\log n$ and at most a small
power of this quantity, this is optimal.  Now if the RIP held, one
would get a squared error bounded by $O(\log n) \frac{s}{m} \sigma^2$
\cite{DS, bickel07} and, therefore, the `RIPless' theory developed in
this paper roughly enjoys the same performance guarantees.

\subsection{A general theory}

The estimate we have just seen is not isolated and the real purpose of
this paper is to develop a theory of compressive sensing which is both
as simple and as general as possible. 

At the heart of compressive sensing is the idea that randomness can be
used as an effective sensing mechanism. We note that random
measurements are not only crucial in the derivation of many
theoretical results, but also generally seem to give better empirical
results as well.  Therefore, we propose a mechanism whereby sensing
vectors are independently sampled from a population
$F$. Mathematically, we observe
 \begin{equation}
  \label{eq:model}
  \tilde{y}_k = \<a_k, x\> +  \sigma z_k, \quad k = 1, \ldots, m,  
\end{equation}
where $x \in \R^n$, $\{z_k\}$ is a noise sequence, and the sensing
vectors $a_k \stackrel{\mathrm{iid}}{\sim} F$. For example, if $F$ is
the family of complex sinusoids, this is the Fourier sampling model
introduced earlier. All we require from $F$ is an isotropy property
and an incoherence property.
\begin{description}
\item {\bf Isotropy property:} We say that $F$ obeys the isotropy
  property if 
  \begin{equation}
    \label{eq:isotropy}
    \E aa^* = \Id, \quad a \sim F.
  \end{equation} 
  If $F$ has mean zero (we do not require this), then $\E a a^*$ is
  the covariance matrix of $F$. In other words, the isotropy condition
  states that the components of $a \sim F$ have unit variance and are
  uncorrelated. This assumption may be weakened a little, as we shall
  see later. 
\end{description}

\begin{description}
\item {\bf Incoherence property:} We may take the coherence parameter
  $\mu(F)$ to be the smallest number such that with $a = (a[1],
  \ldots, a[n]) \sim F$, 
\begin{equation}
\label{eq:coh}
  \max_{1 \leq t \leq n} \abs{a[t]}^2 \leq \mu(F) 
\end{equation}
holds either deterministically or stochastically in the sense
discussed below.  The smaller $\mu(F)$, i.e.~the more incoherent the
sensing vectors, the fewer samples we need for accurate recovery. When
a simple deterministic bound is not available, one can take the
smallest scalar $\mu$ obeying
\begin{equation}
  \label{eq:stochastic_incoherence}
  \E [n^{-1} \twonorm{a}^2 \, \mathbb{1}_{E^c}] \le \tfrac{1}{20} n^{-3/2}  \qquad \text{and} \qquad \P(E^c) \leq (n m)^{-1},  
\end{equation}
where $E$ is the event $\{\max_{1 \leq t \leq n} \abs{a[t]}^2 >
\mu\}$. 
\end{description}

Suppose for instance that the components are
i.i.d.~$\mathcal{N}(0,1)$. Then a simple calculation we shall not
detail shows that
\begin{align}
  \E [n^{-1} \twonorm{a}^2 \, \mathbb{1}_{E^c}] & \le 2n\P(Z > \sqrt{\mu}) + 2\sqrt{\mu} \phi(\sqrt{\mu}),\\
  \P(E^c) & \leq 2 n \P(Z \geq \sqrt{\mu})\nonumber,
\end{align}
where $Z$ is standard normal and $\phi$ is its density function. The
inequality $P(Z > t) \le \phi(t)/t$ shows that one can take $\mu(F)
\le 6\log n$ as long as $n \geq 16$ and $m \leq n$. More
generally, if the components of $a$ are i.i.d.~samples from a
sub-Gaussian distribution, $\mu(F)$ is at most a constant times $\log
n$. If they are i.i.d.~from a sub-exponential distribution, $\mu(F)$
is at most a constant times $\log^2 n$. In what follows, however, it
might be convenient for the reader to assume that the deterministic
bound \eqref{eq:coh} holds.


It follows from the isotropy property that $\E \abs{a[t]}^2 = 1$, and
thus $\mu(F) \ge 1$.  This lower bound is achievable by several
distributions and one such example is obtained by sampling a row from
the DFT matrix as before, so that
\[
a[t] = e^{\i 2\pi k t/n}, 
\]
where $k$ is chosen uniformly at random in $\{0, 1, \ldots,
n-1\}$. Then another simple calculation shows that $\E a a^* = \Id$ and
$\mu(F) = 1$ since $\abs{a[t]}^2 = 1$ for all $t$. At the other
extreme, suppose the measurement process reveals one entry of $x$
selected uniformly at random so that $a = \sqrt{n} \, e_i$ where $i$
is uniform in $\{1, \ldots, n\}$; the normalization ensures that $\E
aa^* = \Id$.  This is a lousy acquisition protocol because one would
need to sample on the order of $n \log n$ times to recover even a
1-sparse vector (the logarithmic term comes from the coupon collector
effect). Not surprisingly, this distribution is in fact highly coherent as $\mu(F) = n$.

With the assumptions set, we now give a representative result of this paper: suppose $x$ is an
arbitrary but fixed $s$-sparse vector and that one collects
information about this signal by means of the random sensing mechanism
\eqref{eq:model}, where $z$ is white noise. Then if the number of
samples is on the order $\mu(F) s \log n$, one can invoke $\ell_1$
minimization to get an estimator $\hat x$ obeying
\[
\|\hat x - x\|_{\ell_2}^2 \le \text{polylog}(n) \, \frac{s}{m} \sigma^2.
\]
This bound is sharp. It is not possible to substantially reduce the
number of measurements and get a similar bound, no matter how
intractable the recovery method might be. Further, with this many
measurements, the upper bound is optimal up to logarithmic
factors. Finally, we will see that when the signal is not exactly
sparse, we just need to add an approximation error to the upper
bound. 

To summarize, this paper proves that one can faithfully recover
approximately $s$-sparse signals from about $s \log n$ random
incoherent measurements for which $\mu(F) = O(1)$.  

\subsection{Examples of incoherent measurements}

We have seen through examples that sensing vectors with low coherence
are global or spread out.  Incoherence alone, however, is not a
sufficient condition: if $F$ were a constant distribution (sampling
from $F$ would always return the same vector), one would not learn
anything new about the signal by taking more samples regardless of the
level of incoherence.  However, as we will see, the incoherence and
isotropy properties together guarantee that sparse vectors lie away
from the nullspace of the sensing matrix whose rows are the $a_k^*$'s.

The role of the isotropy condition is to keep the measurement matrix
from being rank deficient when sufficiently many measurements are
taken (and similarly for subsets of columns of $A$).  Specifically,
one would hope to be able to recover \textit{any} signal from an
arbitrarily large number of measurements.  However, if $\E a a^*$ were
rank deficient, there would be signals $x \in \R^n$ that would not be
recoverable from an arbitrary number of samples; just take $x \neq 0$
in the nullspace of $\E a a^*$.  The nonnegative random variable $x^*
a a^* x$ has vanishing expectation, which implies $a^* x = 0$ almost
surely. (Put differently, all of the measurements would be zero almost
surely.)  In contrast, the isotropy condition implies that
$\frac{1}{m} \sum_{k = 1}^m a_k a_k^* \rightarrow \Id$ almost surely
as $m \rightarrow \infty$ and, therefore, with enough measurements,
the sensing matrix is well conditioned and has a
left-inverse.\footnote{
  One could require `near isotropy,' i.e. $\E a a^* \approx \Id$.  If
  the approximation were tight enough, our theoretical results would
  still follow with minimal changes to the proof.}


We now provide examples of incoherent and isotropic measurements. 

\begin{itemize}
\item {\bf Sensing vectors with independent components.} Suppose the
  components of $a \sim F$ are independently distributed with mean
  zero and unit variance. Then $F$ is isotropic. In addition, if the
  distribution of each component is light-tailed, then the
  measurements are clearly incoherent.

  A special case concerns the case where $a \sim N(0, \Id)$, also
  known in the field as the {\em Gaussian measurement ensemble}, which
  is perhaps the most commonly studied. Here, one can take $\mu(F)
    = 6\log n$ as seen before.

  Another special case is the {\em binary measurement ensemble} where
  the entries of $a$ are symmetric Bernoulli variables taking on the
  values $\pm 1$. A shifted version of this distribution is the
  sensing mechanism underlying the single pixel camera
  \cite{duarte08}.

\item {\bf Subsampled orthogonal transforms:} Suppose we have an
  orthogonal matrix obeying $U^* U = n \, \Id$. Then consider the
  sampling mechanism picking rows of $U$ uniformly and independently
  at random. In the case where $U$ is the DFT, this is the random
  frequency model introduced earlier. Clearly, this distribution is
  isotropic and $\mu(F) = \max_{ij} |U_{ij}|^2$. In the case where $U$
  is a Hadamard matrix, or a complex Fourier matrix, $\mu(F) = 1$.

\item {\bf Random convolutions:} Consider the circular convolution
  model $y = G x$ in which
\[
 G = \begin{bmatrix} g[0] & g[1] & g[2] & \ldots & g[n-1] \\
g[n-1] & g[0]  & g[1] & \ldots & \\
 &  &  &  & \\
 &  &  &  &\\
g[1] & & \ldots & g[n-1] & g[0]
\end{bmatrix}. 
\]
%
Because a convolution is diagonal in the Fourier domain (we just
multiply the Fourier components of $x$ with those of $g$), $G$ is an
isometry if the Fourier components of $g = (g[0], \ldots, g[n-1])$
have the same magnitude. In this case, sampling a convolution product
at randomly selected time locations is an isotropic and incoherent
process provided $g$ is spread out ($\mu(F) = \max_t
\abs{g(t)}^2$). This example extends to higher dimensions; e.g.~to
spatial 3D convolutions.

\item {\bf Subsampled tight or continuous frames:} We can generalize
  the example above by subsampling a tight frame or even a continuous
  frame. An important example might be the Fourier transform with a
  continuous frequency spectrum.  Here,
\[
a(t) = e^{\i 2\pi \omega t},
\]
where $\omega$ is chosen uniformly at random in $[0,1]$ (instead of
being on an equispaced lattice as before). This distribution is
isotropic and obeys $\mu(F) = 1$. A situation where this arises is in
magnetic resonance imaging (MRI) as frequency samples rarely fall on
an equispaced Nyquist grid. By swapping time and frequency, this is
equivalent to sampling a nearly sparse trigonometric polynomial at
randomly selected time points in the unit interval
\cite{rauhut2007random}.
\end{itemize}
These examples could of course be multiplied, and we hope we have made
clear that our framework is general and encompasses all the
measurement models commonly discussed in compressive sensing --- and
perhaps many more.

\subsection{Matrix notation}

Before continuing, we pause to introduce some useful matrix
notation. Divide both sides of \eqref{eq:model} by $\sqrt{m}$, and
rewrite our statistical model as
\begin{equation}
\label{eq:modelmatrix}
y = Ax + \sigma_m z; 
\end{equation}
the $k$th entry of $y$ is $\tilde{y}_k$ divided by $\sqrt{m}$, the
$k$th row of $A$ is $a_k^*$ divided by $\sqrt{m}$, and $\sigma_m$ is
$\sigma$ divided by $\sqrt{m}$.  This normalization implies that the
columns of $A$ are approximately unit-normed, and is most used in the
compressive sensing literature. 

\subsection{Incoherent sampling theorem}

For pedagogical reasons, we introduce our results by first presenting
a recovery result from noiseless data. The recovered signal is
obtained by the standard $\ell_1$-minimization program
\begin{equation}
\label{eq:el1}
\min_{\bar x \in \R^n}  \,\,\, \onenorm{\bar{x}} \quad \text{subject to} 
\quad A \bar{x} = y. 
\end{equation}
(Recall that the rows of $A$ are normalized independent samples from
$F$.) 

\begin{theorem}[Noiseless incoherent sampling]
\label{teo:noiseless}
Let $x$ be a fixed but otherwise arbitrary $s$-sparse vector in
$\R^n$. Then with probability at least $1 - 5/n -  e^{-\beta}$, $x$
is the unique minimizer to \eqref{eq:el1} with $y = Ax$ provided that
\[
m \ge C_\beta \cdot \mu(F) \cdot s \cdot \log n.
\]
More precisely, $C_\beta$ may be chosen as $C_0 (1+\beta)$ for some
positive numerical constant $C_0$.
\end{theorem}
Among other things, this theorem states that one can perfectly recover
an arbitrary sparse signal from about $s \log n$ convolution samples,
or a signal that happens to be sparse in the wavelet domain from about
$s \log n$ randomly selected noiselet coefficients. It extends an
earlier result \cite{CR07}, which assumed a subsampled orthogonal
model, and strengthens it since that reference could only prove the
claim for randomly signed vectors $x$. Here, $x$ is arbitrary, and we
do not make any distributional assumption about its support or its
sign pattern.

This theorem is also about a fundamental information theoretic
limit: the number of samples for perfect recovery has to be on the
order of $\mu(F) \cdot s \cdot \log n$, and cannot possibly be much
below this number.  More precisely, suppose we are given a
distribution $F$ with coherence parameter $\mu(F)$.  Then there exist
$s$-sparse vectors that cannot be recovered with probability at least
$1-1/n$, say, from fewer than a constant times $\mu(F) \cdot s \cdot
\log n$ samples. When $\mu(F) = 1$, this has been already established
since \cite{CRT06} proves that some $s$ sparse signals cannot be
recovered from fewer than a constant times $s \cdot \log n$ random DFT
samples. Our general claim follows from a modification of the argument
in \cite{CRT06}. Assume, without loss of generality, that $\mu(F)$ is an
integer and consider the isotropic process that samples rows from an
$n \times n$ block diagonal matrix, each block being a DFT of a
smaller size; that is, of size $n/\ell$ where $\mu(F) = \ell$. Then if
$m \le c_0 \cdot \mu(F) \cdot s \cdot \log n$, one can construct
$s$-sparse signals just as in \cite{CRT06} for which $Ax = 0$ with
probability at least $1/n$. We omit the details.

The important aspect, here, is the role played by the coherence
parameter $\mu(F)$. In general, the minimal number of samples must be
on the order of the coherence times the sparsity level $s$ times a
logarithmic factor. Put differently, {\em the coherence completely
  determines the minimal sampling rate}.

\subsection{Main results}
\label{sec:main-result}

We assume for simplicity that we are undersampling so that $m \leq n$.
Our general result deals with 1) arbitrary signals which are not
necessarily sparse (images are never exactly sparse even in a
transformed domain) and 2) noise.
To recover $x$ from the data $y$ and the model \eqref{eq:modelmatrix}, we
consider the unconstrained LASSO which solves the $\ell_1$ regularized
least-squares problem
\begin{equation}
\label{eq:Lasso}
\min_{\bar x \in \R^n} \, \, \, \tfrac{1}{2} \twonorm{A \bar{x} - y}^2 + 
\lambda \sigma_m \onenorm{\bar{x}}. 
\end{equation}
We assume that $z$ is Gaussian $z \sim N(0, \Id)$. However, the
theorem below remains valid as long as $\infnorm{A^*z} \leq \lambda_n$
for some $\lambda_n \ge 0$, and thus many other noise models would
work as well. In what follows, $x_s$ is the best $s$-sparse
approximation of $x$ or, equivalently, the vector consisting of the
$s$ largest entries of $x$ in magnitude.

\begin{theorem}
\label{teo:noisy}
Let $x$ be an arbitrary fixed vector in $\R^n$. Then with probability
at least $1 - 6/n - 6 e^{-\beta}$ the solution to \eqref{eq:Lasso}
with $\lambda = 10 \sqrt{\log n}$ obeys
\begin{equation}
\label{eq:twoBound}
\twonorm{\hat{x} - x}  \leq \min_{1 \leq s \leq \bar{s}} C (1 + \alpha) \,  \left[\frac{\onenorm{x - x_s}}{\sqrt{s}} + \sigma \sqrt{\frac{s\log n}{m}}\right] 
\end{equation}
provided that $m \ge C_\beta \cdot \mu(F) \cdot \bar s \cdot \log n$.
If one measures the error in the $\ell_1$ norm, then
\begin{equation}
\label{eq:oneBound}
\onenorm{\hat{x} - x} \leq \min_{1 \leq s \leq \bar{s}} C (1 + \alpha) \left[\onenorm{x - x_s} + s \sigma \sqrt{\frac{\log n}{m}}\right]. 
\end{equation}
Above, $C$ is a numerical constant, $C_\beta$ can be chosen as before,
and $\alpha = \sqrt{\frac{(1 + \beta) s \log^5 n}{m}}$ which is never
greater than $\log^2 n$ in this setup.
\end{theorem}
These robust error bounds do not require either (1) a random model on
the signal or (2) the RIP nor one of a few closely related strong
conditions such as the RIP-1 \cite{gilbert2010sparse}, the restricted
eigenvalue assumption \cite{bickel07} or the compatibility condition
\cite{van2009conditions}. The conditions are weak enough that they do
not necessarily imply uniform sparse-signal recovery, but instead they
imply recovery of an arbitrary \textit{fixed} sparse signal with high
probability.  Further, the error bound is within at most a $\log^2 n$
factor of what has been established using the RIP since a variation on
the arguments in \cite{DS} would give an error bound proportional to
the quantity inside the square brackets in \eqref{eq:twoBound}. As a
consequence, the error bound is within a polylogarithmic factor of
what is achievable with the help of an oracle that would reveal the
locations of the significant coordinates of the unknown signal
\cite{DS}. In other words, it cannot be substantially improved.

Because much of the compressive sensing literature works with
restricted isometry conditions -- we shall discuss exceptions such as
\cite{donoho2010noise, bayati2010lasso} in Section
\ref{sec:contribution} -- we pause here to discuss these conditions
and to compare them to our own.  We say that an $m \times n$ matrix
$A$ obeys the RIP with parameters $s$ and $\delta$ if
\begin{equation}
  \label{eq:rip}
  (1-\delta) \|v\|_{\ell_2}^2 \le \|Av\|_{\ell_2}^2 \le (1+\delta) 
\|v\|_{\ell_2}^2  
\end{equation}
for all $s$-sparse vectors $v$. In other words, all the submatrices of
$A$ with at most $s$ columns are well conditioned. When the RIP holds
with parameters $2s$ and $\delta < 0.414\ldots$ \cite{RIP} or even
$\delta \le 0.453\ldots$ \cite{FoucartLai}, it is known that the error
bound \eqref{eq:twoBound} holds (without the factor
$(1+\alpha)$). This $\delta$ is sometimes referred to as the
restricted isometry constant.

Bounds on the restricted isometry constant have been established in
\cite{CT06} and in \cite{rudelsonVershynin} for partial DFT matrices,
and by extension, for partial subsampled orthogonal transforms. For
instance, \cite{rudelsonVershynin} proves that if $A$ is a properly
normalized partial DFT matrix, then the RIP with $\delta = 1/4$ holds
with high probability if $m \geq C \cdot s \log^4 n$ ($C$ is some
positive constant). We believe the proof extends with hardly any
change to show that the measurement ensembles considered in this paper
obey the RIP with high probability when $m \geq C \cdot \mu(F) \cdot s
\log^4 n$.  Thus, our result bridges the gap between the region where
the RIP holds and the region in which one has the minimum number of
measurements needed to prove perfect recovery of exactly sparse
signals from noisy data, which is on the order of $\mu(F) \cdot s \log
n$.

The careful reader will no doubt remark that for very specific models
such as the Gaussian measurement ensemble, it is known that on the
order $s \log n/s$ samples are sufficient for stable recovery while
our result asserts that on the order of $s \log^2 n$ are sufficient
(and $s \log n$ for the binary measurement ensemble). This slight loss
is a small price to pay for a very simple general theory, which
accommodates a wide array of sensing strategies. Having said this, the
reader will also verify that specializing our proofs below gives an
optimal result for the Gaussian ensemble; i.e.~establishes a
near-optimal error bound from about $s \log n/s$ observations.

Finally, another frequently discussed algorithm for sparse regression
is the Dantzig selector \cite{DS}. Here, the estimator is given by the
solution to the linear program
\begin{equation}
\label{eq:DS}
\min_{\bar x \in \R^n} \,\,\, \onenorm{\bar{x}} \quad \text{subject to} \quad 
\infnorm{A^*(A \bar{x} - y)} \leq \lambda \, \sigma_m.  
\end{equation}
We show that the
Dantzig selector obeys nearly the same error bound.
\begin{theorem}
\label{teo:noisyDS}
The Dantzig selector, with $\lambda = 10 \sqrt{\log n}$ and everything
else the same as in Theorem \ref{teo:noisy}, obeys
\begin{align}
\label{eq:twoBoundDS}
\twonorm{\hat{x} - x} & \leq \min_{s \leq \bar s} C (1 + \alpha^2)
\, \left[\frac{\onenorm{x - x_s}}{\sqrt{s}} + \sigma \sqrt{\frac{s\log
      n}{m}}\right]\\
\label{eq:oneBoundDS}
\onenorm{\hat{x} - x} & \leq \min_{s \leq \bar s} C  (1 + \alpha^2) \left[\onenorm{x - x_s} + s \sigma \sqrt{\frac{\log n}{m}}\right]
\end{align}
with the same probabilities as before. 
\end{theorem}
The only difference is $\alpha^2$ instead of $\alpha$ in the
right-hand sides.

\subsection{Our contribution}
\label{sec:contribution}

The main contribution of this paper is to provide a simple framework
which applies to all the standard compressive sensing models and some
new ones as well. The results in this paper also reduce the minimal
number of measurements theoretically required in some standard sensing
models such as Fourier measurements, or, more generally, sensing
matrices obtained by sampling a few rows from an orthogonal matrix.
We establish that the restricted isometry property is not necessarily
needed to accurately recover nearly sparse vectors from noisy
compressive samples. This may be of interest because in many
situations, the RIP may be hard to check, or may not hold, or does not
hold. Thus our work is a significant departure from the majority of
the literature, which establishes good noisy recovery properties via
the RIP machinery. This literature is, of course, extremely large and
we cannot cite all contributions but a partial list would include
\cite{CT06,most-large,rudelsonVershynin,CRT2,DS,bickel07,CDD,wojtaszczyk08,zhang2009some,baraniuk2008simple,
raginsky2010performance,jafarpour2009efficient,cai2010new,bah2010improved}.

The reason why one can get strong error bounds, which are within a
polylogarithmic factor of what is available with the aid of an
`oracle,' without the RIP is that our results do not imply
universality. That is, we are not claiming that if $A$ is randomly
sampled and then fixed once for all, then the error bounds from
Section \ref{sec:main-result} hold for all signals $x$. What we are
saying is that if we are given an arbitrary $x$, and then collect data
by applying our random scheme, then the recovery of {\em this} $x$
will be accurate.  

If one wishes to establish universal results holding for {\em all} $x$
simultaneously, then we would need the RIP or a property very close to
it. As a consequence, we cannot possibly be in this setup and
guarantee universality since we are not willing to assume that the RIP
holds.  To be sure, suppose we had available an oracle informing us
about the support $T$ of $x$. Then we would need the pseudo-inverse of
the submatrix with columns in $T$ to be bounded. In other words, the
minimum singular value of this submatrix would have to be away from
zero.  For a universal result, this would need to be true for all
subsets of cardinality $s$; that is, the minimum singular value of all
submatrices with $s$ columns would have to be away from zero. This
essentially is the restricted isometry property.

To the best of our knowledge, only a few other papers have addressed
non-universal stability (the literature grows so rapidly that an
inadvertent omission is entirely possible). In an earlier work
\cite{CP07}, the authors also considered weak conditions that allow
stable recovery; in this case the authors assumed that the signal was
sampled according to a random model, but in return the measurement
matrix $A$ could be deterministic.  In the asymptotic case, stable
signal recovery has been demonstrated for the Gaussian measurement
ensemble in a regime in which the RIP does not necessarily hold
\cite{donoho2010noise, bayati2010lasso}.  In fact, the authors of
\cite{donoho2010noise, bayati2010lasso} are able to give exact limits
on the error rather than error bounds.  Aside from these papers and
the work in progress \cite{ripless}, it seems that that the literature
regarding stable recovery with conditions weak enough that they do not
imply universality is extremely sparse.  Finally and to be complete,
we would like to mention that earlier works have considered the
recovery of perfectly sparse signals from subsampled orthogonal
transforms \cite{CR07}, and of sparse trigonometric polynomials from
random time samples \cite{rauhut2007random}.

\subsection{Organization of the paper}

The paper is organized as follows.  In Section \ref{sec:estimates}, we
introduce several fundamental estimates which our arguments rely upon,
but which also could be useful tools for other results in the field.
In Section 3, we prove the noiseless recovery result, namely, Theorem
\ref{teo:noiseless}. In Section \ref{sec:noisy}, we prove our main
results, Theorems \ref{teo:noisy} and \ref{teo:noisyDS}.  Now all
  these sections assume for simplicity of exposition that the
  coherence bound holds deterministically \eqref{eq:coh}. We extend
  the proof to distributions obeying the coherence property in the
  stochastic sense \eqref{eq:stochastic_incoherence} in the
  Appendix. This Appendix also contains another important technical
  piece, namely, a difficult proof of an intermediate result (weak RIP
  property).  Finally, we conclude the main text with some final
  comments in Section \ref{sec:discussion}.


\subsection{Notation}

We provide a brief summary of the notations used throughout the
paper. For an $m \times n$ matrix $A$ and a subset $T \subset \{1, \ldots,
n\}$, $A_T$ denotes the $m \times |T|$ matrix with column indices in
$T$. Also, $A_{\{i\}}$ is the $i$-th column of $A$. Likewise, for a
vector $v \in \R^n$, $v_T$ is the restriction of $v$ to indices in
$T$. Thus, if $v$ is supported on $T$, $Av = A_T v_T$.  In particular,
$a_{k,T}$ is the vector $a_k$ restricted to $T$.  The operator norm of
a matrix $A$ is denoted $\opnorm{A}$.  The identity matrix, in any
dimension, is denoted $\Id$. Further, $e_i$ always refers to the
$i$-th standard basis element, e.g., $e_1 = (1, 0, \ldots, 0)$. For a
scalar $t$, $\sgn{t}$ is the sign of $t$ if $t \neq 0$ and is zero
otherwise. For a vector $x$, $\sgn{x}$ applies the sign function
componentwise.  We shall also use $\mu$ as a shorthand for $\mu(F)$
whenever convenient.  Throughout, $C$ is a constant whose value may
change from instance to instance.

\section{Fundamental Estimates}
\label{sec:estimates}

Our proofs rely on several estimates, and we provide an interpretation
of each whenever possible.  The first estimates {\bf E1}--{\bf E4} are
used to prove the noiseless recovery result; when combined with the
weak RIP, they imply stability and robustness.  Throughout this
section, $\delta$ is a parameter left to be fixed in later sections;
it is always less than or equal to one.

\subsection{Local isometry}
\label{sec:estimate1}

Let $T$ of cardinality $s$ be the support of $x$ in Theorem
\ref{teo:noiseless}, or the support of the best $s$-sparse
approximation of $x$ in Theorem \ref{teo:noisy}. We shall need that
with high probability, 
\begin{equation}
\label{eq:localConditioning}
\opnorm{A^*_T A_T - \Id} \leq \delta
\end{equation}
with $\delta \leq 1/2$ in the proof of Theorem \ref{teo:noiseless} and
$\delta \leq 1/4$ in that of Theorem \ref{teo:noisy}. Put differently,
the singular values of $A_T$ must lie away from zero. This condition
essentially prevents $A_T$ from being singular as, otherwise, there
would be no hope of recovering our sparse signal $x$. Indeed, letting
$h$ be any vector supported on $T$ and in the null space of $A$, we
would have $A x = A(x+h)$ and thus, recovery would be impossible even
if one knew the support of $x$. The condition
\eqref{eq:localConditioning} is much weaker than the restricted
isometry property because it does not need to hold uniformly over all
sparse subsets --- only on the support set.
\begin{lemma}[E1: local isometry]
\label{lem:estimate1}
\label{lem:localConditioning}
Let $T$ be a fixed set of cardinality $s$. Then for $\delta > 0$,
 \begin{equation}
   \P\left(\opnorm{A^*_T A_T - I} \geq \delta\right) \leq 2 s \exp\left(-\frac{m}{\mu(F) s} \,\cdot\, \frac{\delta^2}{2(1 + \delta/3)}\right).
 \end{equation} 
 In particular, if $m \geq \frac{56}{3} \mu(F) \cdot s \cdot \log n$, then
\[\P\left(\opnorm{A^*_T A_T - I} \geq 1/2 \right) \leq 2/n.\]
\end{lemma}
Note that $\opnorm{A^*_T A_T - \Id} \leq \delta$ implies that
$\opnorm{(A^*_T A_T)^{-1}} \leq 1/(1 - \delta)$, a fact that we will
use several times.

In compressive sensing, the classical way of proving such estimates is
via Rudelson's selection theorem \cite{rudelson99}. Here, we use a
more modern technique based on the matrix Bernstein inequality of
Ahlswede and Winter \cite{ahlswede00}, developed for this setting by
Gross \cite{gross09}, and tightened in \cite{tropp10} by Tropp and in
\cite{oliveira10} by Oliveira.  We present the version in
\cite{tropp10}.

\begin{theorem}[Matrix Bernstein inequality] Let $\{X_k\} \in
  \R^{d\times d}$ be a finite sequence of independent random
  self-adjoint matrices.  Suppose that $\E X_k = 0$ and $\opnorm{X_k}
  \leq B$ a.s.~and put
  \[
  \sigma^2 \coloneq \opnorm{\sum_k \E X_k^2}.
  \] 
  Then for all $t \geq 0$,
\begin{equation}
\label{eq:bernstein}
\P\left(\opnorm{\sum_k X_k} \geq t\right) \leq 2 d \exp\left(\frac{-t^2/2}{\sigma^2 + Bt/3}\right).
\end{equation}
\end{theorem}

\begin{proof} Decompose $A^*_T A_T - \Id$ as
\[
A^*_T A_T - \Id = m^{-1} \sum_{k=1}^m (a_{k,T} a_{k,T}^* - \Id) =
m^{-1} \sum_{k=1}^m X_k, \qquad X_k \coloneq a_{k,T} a_{k,T}^* - \Id.
\] 
The isotropy condition implies $\E X_k = 0$, and since
$\twonorm{a_{T}}^2 \le \mu(F) \cdot s$, we have $\opnorm{X_k} =
\max(\twonorm{a_{i,T}}^2 - 1, 1) \leq \mu(F) \cdot s$. Lastly, $ 0
\preceq \E X_k^2 = \E (a_{k, T} a_{k, T}^*)^2 - \Id \preceq \E (a_{k,
  T} a_{k, T}^*)^2 = \E \|a_{k, T}\|^2 a_{k, T} a_{k, T}^*$. However,
\[
\E \|a_{k, T}\|^2 a_{k, T} a_{k, T}^* \preceq \mu(F) \cdot s \cdot \E
a_{k, T} a_{k, T}^* = \mu(F) \cdot s \cdot \Id
\]
and, therefore, $\sum_k \E X_k^2 \preceq m \cdot \mu(F) \cdot s \cdot
I$ so that $\sigma^2$ is bounded above by $m \cdot \mu(F) \cdot s$.
Plugging $t = \delta m$ into \eqref{eq:bernstein} gives the lemma.
\end{proof}

Instead of having $A$ act as a near isometry on all vectors supported
on $T$, we could ask that it preserves the norm of an arbitrary fixed
vector (with high probability), i.e.~$\twonorm{A v} \approx
\twonorm{v}$ for a fixed $v$ supported on $T$.  
Not surprisingly, this can be proved with generally (slightly) weaker
requirements.
\begin{lemma}[E2: low-distortion]
  \label{lem:estimate2} Let $v$ be a fixed vector supported on a set
  of cardinality at most $s$.  Then for each $t \leq 1/2$,
\[
\P(\twonorm{(A^* A - I)v} \geq t \twonorm{v}) \leq \exp
\Bigl(-\frac{1}{4} \Bigl(t \sqrt{\frac{m}{\mu(F) s}} - 1\Bigr)^2
\Bigr).
\]
\end{lemma}    
The proof is an application of the vector Bernstein inequality
described in the fourth estimate {\bf E4}. It is analogous to the
proof shown there and is not repeated.

\subsection{Off-support incoherence}
\label{sec:estimate2}

\begin{lemma}[E3: off-support incoherence]
\label{lem:estimate3}
Let $v$ be supported on $T$ with $|T| = s$. Then for each $t > 0$, 
\begin{equation}
  \P(\infnorm{A_{T^c}^* A v} \geq t \twonorm{v}) \leq 2 n \exp\left(- \frac{m}{2\mu(F)} \,\cdot\, \frac{t^2}{1 + \frac{1}{3} \sqrt{s} t}\right). 
\end{equation}
\end{lemma}
This lemma says that if $v = x$, then $\max_{i \in T^c} |\<A_{\{i\}}, Ax\>|$
cannot be too large so that the off-support columns do not correlate
too well with $Ax$. The proof of {\bf E3} is an application of
Bernstein's inequality --- the matrix Bernstein inequality with
$d=1$ --- together with the union bound. 

\begin{proof} We have
\[
\infnorm{A_{T^c}^* A v} = \max_{i \in T^c} \abs{\<e_i, A^* A v\>}.
\]
Assume without loss of generality that $\twonorm{v} = 1$, fix $i \in
T^c$ and write
\[
\<e_i, A^* A v\> = \frac{1}{m} \sum_k g_k, \qquad g_k  \coloneq
\<e_i, a_k a_{k}^* v\>. 
\]
Since $i \in T^c$, $\E g_k = 0$ by the isotropy property. Next, the
Cauchy-Schwartz inequality gives $\abs{g_k} = \abs{\<e_i, a_k\> \cdot
  \<a_{k}, v\>} \leq \abs{\<e_i, a_k\>} \twonorm{a_{k,T}}$. Since $
\abs{\<e_i, a_k\>} \le \sqrt{\mu(F)}$ and $ \twonorm{a_{k,T}} \le
\sqrt{\mu(F)} s$, we have $|g_k| \le \mu(F) \sqrt{s}$. Lastly, for the
total variance, we have
\[
\E g_k^2 \leq \mu(F) \E\<a_{k,T}, v\>^2 = \mu(F)
\]
where the equality follows from the isotropy property.  Hence,
$\sigma^2 \leq m \mu(F)$, and Bernstein's inequality gives
\[
\P(\abs{\<e_i, A^* A v\>} \geq t) \leq 2\exp\left(- \frac{m}{2\mu(F)}\,\cdot
  \, \frac{t^2}{1 + \frac{1}{3} \sqrt{s} t}\right).
\]
 Combine this with
the union bound over all $i \in T^c$ to give the desired result.
\end{proof}

We also require the following related bound:
\[
\max_{i \in T^c} \twonorm{A^*_T A_{\{i\}}} \leq \delta.
\]
In other words, none of the column vectors of $A$ outside of the
support of $x$ should be well approximated by \textit{any} vector
sharing the support of $x$.
\begin{lemma}[E4: uniform off-support incoherence]
\label{lem:estimate4}
Let $T$ be a fixed set of cardinality $s$. For any $0 \leq t \leq
\sqrt{s}$, 
\[
\P\left(\max_{i \in T^c} \twonorm{A^*_T A_{\{i\}}} \geq t\right) \leq
n \exp\left(-\frac{m t^2}{8 \mu(F) s} + \frac{1}{4}\right).
\] 
In particular, if $m \geq 8 \mu(F) \cdot s \cdot (2 \log n + 1/4)$,
then
\[
\P\left(\max_{i \in T^c} \twonorm{A^*_T A_{\{i\}}} \geq 1\right) \leq 1/n.
\]
\end{lemma} 

The estimate follows from the vector Bernstein inequality, proved by
Gross \cite[Theorem 11]{gross09}.  We use a slightly weaker version,
which we find slightly more convenient.
\begin{theorem}[Vector Bernstein inequality] Let $\{v_k\} \in \R^{d}$
  be a finite sequence of independent random vectors.  Suppose that
  $\E v_k = 0$ and $\twonorm{v_k} \leq B$ a.s.~and put $\sigma^2 \geq
  \sum_k \E\twonorm{v_k}^2$.  Then for all $0 \leq t \leq \sigma^2/B$,
\begin{equation}
\label{eq:vecBernstein}
\P\left(\twonorm{\sum_k v_k} \geq t\right) \leq \exp\left(-\frac{(t/\sigma-1)^2}{4}\right) \leq \exp\left(-\frac{t^2}{8\sigma^2} + \frac{1}{4}\right).
\end{equation}
\end{theorem}
Note that the bound does not depend on the dimension $d$. 

\begin{proof} Fix $i \in T^c$ and write
\[
A^*_T A_{\{i\}} = \frac{1}{m} \sum_{j=1}^m a_{k,T}^* \<a_k, e_i\> \coloneq \frac{1}{m} \sum_{k=1}^m v_k.
\]
As before, $\E v_k = \E a_{k,T}^* \<a_k, e_i\> = 0$ since $i \in T^c$.
Also, $\twonorm{v_k} = \twonorm{a_{k,T}} \abs{\<a_k, e_i\>} \leq
\mu(F) \sqrt{s}$. Lastly, we calculate the sum of expected squared
norms,
\[
\sum_{k=1}^m \E \twonorm{v_k}^2 = m \E\twonorm{v_1}^2 \leq m
\E[\twonorm{a_{1,T}}^2 \<e_i, a_1\>^2] \leq m \mu(F) s \cdot \E \<e_i,
a_1\>^2 = m \mu(F) s.
\]
As before, the last equality follows from the isotropy property.
Bernstein's inequality together with the union bound give the lemma.
\end{proof}

\subsection{Weak RIP}
\label{sec:weak-RIP}

In the nonsparse and noisy setting, we shall make use of a variation
on the restricted isometry property to control the size of the
reconstruction error. This variation is as follows: 
\begin{theorem}[E5: weak RIP]
\label{teo:weakRIP}
Let $T$ be a fixed set of cardinality $s$ and fix $\delta > 0$. Then
for all $v$ supported on $T \cup R$, where $R$ is any set of
cardinality $|R| \le r$, we have 
\begin{equation}
\label{eq:weakRIP}
(1 - \delta) \twonorm{v}^2 \leq \twonorm{A v}^2 \leq (1 + \delta) \twonorm{v}^2
\end{equation} 
with probability at least $1 - 5 e^{-\beta}$ provided that
\[
m \geq C_\delta \cdot \beta \cdot \mu(F) \cdot \max(s \log n, r \log^5
n).
\]
Here $C_\delta$ is a fixed numerical constant which only depends upon
$\delta$. 
\end{theorem}
This theorem is proved in the Appendix using Talagrand's majorizing measures theorem, and
combines the framework and results of Rudelson and Vershynin in
\cite{rudelsonVershynin} and \cite{rudelson99}. In the proof of
Theorem \ref{teo:noisy}, we take $\delta = 1/4$.

The condition says that the column space of $A_T$ should not be too
close to that spanned by another small disjoint set $R$ of columns. To
see why a condition of this nature is necessary for any recovery
algorithm, suppose that $x$ has fixed support $T$ and that there is a
single column $A_{\{i\}}$ which is a linear combination of columns in
$T$, i.e., $A_{T \cup \{i\}}$ is singular.  Let $h \neq 0$ be
supported on $T \cup \{i\}$ and in the null space of $A$.  Then $A x =
A (x + t h)$ for any scalar $t$.  Clearly, there are some values of
$t$ such that $x + t h$ is at least as sparse as $x$, and thus one
should not expect to be able to recover $x$ by any method. In general,
if there were a vector $v$ as above obeying $\twonorm{A v} \ll
\twonorm{v}$ then one would have $A_T v_T \approx -A_R v_R$.  Thus, if
the signal $x$ were the restriction of $v$ to $T$, it would be very
difficult to distinguish it from that of $-v$ to $R$ under the
presence of noise.


The weak RIP is a combination of the RIP and the local conditioning
estimate {\bf E1}. When $r = 0$, this is {\bf E1} whereas this is the
restricted isometry property when $s = 0$. The point is that we do not
need the RIP to hold for sparsity levels on the order of $m/[\mu(F)
\log n]$.  Instead we need the following property: consider an
arbitrary submatrix formed by concatenating columns in $T$ with $r$
other columns from $A$ selected in any way you like; then we would
like this submatrix to be well conditioned. Because $T$ is fixed, one
can prove good conditioning when $s$ is significantly larger than the
maximum sparsity level considered in the standard RIP.

\subsection{Implications}

The careful reader may ask why we bothered to state estimates {\bf
  E1}--{\bf E4} since they are all implied by the weak RIP!  Our
motivation is three-fold: (1) some of these estimates, e.g.~{\bf E2}
hold with better constants and weaker requirements than those implied
by the weak RIP machinery; (2) the weak RIP requires an in-depth proof
whereas the other estimates are simple applications of well-known
theorems, and we believe that these theorems and the estimates should
be independently useful tools to other researchers in the field; (3)
the noiseless theorem does not require the weak RIP.

\section{Noiseless and Sparse Recovery}
\label{sec:noiseless}

This section proves the noiseless recovery theorem, namely, Theorem
\ref{teo:noiseless}. Our proof essentially adapts the arguments of
David Gross \cite{gross09} from the low-rank matrix recovery
problem.
\subsection{Dual certificates}

The standard method for establishing exact recovery is to exhibit a
{\em dual certificate}; that is to say, a vector $v$ obeying the two
properties below.
\begin{lemma}[Exact duality]
  Set $T = \text{supp}(x)$ with $x$ feasible for \eqref{eq:el1}, and
  assume $A_T$ has full column rank.  Suppose there exists $v \in
  \R^n$ in the row space of $A$ obeying
\begin{equation}
  v_T = \sgn{x_T} \qquad \text{and} \qquad \infnorm{v_{T^c}} < 1.
\end{equation}
Then $x$ is the unique $\ell_1$ minimizer to \eqref{eq:el1}.
\end{lemma}
The proof is now standard, see \cite{CT06}.  Roughly, the existence of
a dual vector implies that there is a subgradient of the $\ell_1$ norm
at $x$ that is perpendicular to the feasible set. This geometric
property shows that $x$ is solution. Following Gross, we slightly
modify this definition as to make use of an `inexact dual vector.'

\begin{lemma}[Inexact duality]
\label{lem:inexactDuality}
Set $T = \text{supp}(x)$ where $x$ is feasible, and assume that
\begin{equation}
\label{eq:inexactConditions}
\opnorm{(A^*_T A_T)^{-1}} \leq 2 \qquad \text{and} \qquad \max_{i \in T^c}\twonorm{A^*_T A_{\{i\}}} \leq 1.
\end{equation}
Suppose there exists $v \in \R^n$ in the row space of $A$ obeying
\begin{equation}
  \twonorm{v_T - \sgn{x_T}} \leq 1/4 \qquad \text{and} \qquad 
  \infnorm{v_{T^c}} \leq 1/4.
\end{equation}
Then $x$ is the unique $\ell_1$ minimizer to \eqref{eq:el1}.
\end{lemma}
\begin{proof} Let $\hat{x} = x + h$ be a solution to \eqref{eq:el1}
  and note that $Ah = 0$ since both $x$ and $\hat x$ are feasible. To
  prove the claim, it suffices to show that $h = 0$.  We begin by
  observing that
\[
\onenorm{\hat{x}} = \onenorm{x_T + h_T} + \onenorm{h_{T^c}} \ge
\onenorm{x_T} + \<\sgn{x_T},h_T\> + \onenorm{h_{T^c}}.
\]
Letting $v = A^* w$ be our (inexact) dual vector, we have
\[
\<\sgn{x_T},h_T\> = \<\sgn{x_T} - v_T, h_T\> + \<v_T, h_T\> =
\<\sgn{x_T} - v_T, h_T\> - \<v_{T^c}, h_{T^c}\>,
\]
where we used $\<v_T, h_T\> = \<v,h\> - \<v_{T^c}, h_{T^c}\> = -
\<v_{T^c}, h_{T^c}\>$ since $\<v, h\> = \<w, Ah\> = 0$. The
Cauchy-Schwartz inequality together with the properties of $v$ yield
\[
|\<\sgn{x_T},h_T\>| \le \frac{1}{4} (\twonorm{h_T} + \onenorm{h_{T^c}})
\]
and, therefore, 
\[
\onenorm{\hat{x}} \ge \onenorm{x} - \frac{1}{4} \twonorm{h_T} +
\frac{3}{4} \onenorm{h_{T^c}}.
\]

We now bound $\twonorm{h_T}$.  First, it follows from
\[
h_T = (A^*_T A_T)^{-1} A^*_T A_T h = -(A^*_T A_T)^{-1} A^*_T A_{T^c} h_{T^c}
\]
that $\twonorm{h_T} \le 2 \twonorm{A^*_T A_{T^c} h_{T^c}}$. Second,
\[
\twonorm{A^*_T A_{T^c} h_{T^c}} \leq 2 \sum_{i \in {T^c}}
\twonorm{A^*_T A_{\{i\}}} \abs{h_i} \le \max_{i \in T^c}
\twonorm{A^*_T A_{\{i\}}} \onenorm{h_{T^c}} \le \onenorm{h_{T^c}}. 
\]
In conclusion, $\|h_T\|_2 \le 2 \|h_{T^c}\|_1$ and thus,
\[
\onenorm{\hat{x}} \geq \onenorm{x} + \frac{1}{4} \onenorm{h_{T^c}}.
\]
This implies $h_{T^c} = 0$, which in turn implies $h_T = 0$ since we
must have $A_T h_T = Ah = 0$ (and $A_T$ has full rank).
\end{proof}

\begin{lemma}[Existence of a dual certificate]
\label{lem:dualVector}
Under the hypotheses of Theorem \ref{teo:noiseless}, one can find $v
\in \R^n$ obeying the conditions of Lemma \ref{lem:inexactDuality}
with probability at least $1 - e^{-\beta} - 1/n$.
\end{lemma}

This lemma, which is proved next, implies Theorem
\ref{teo:noiseless}. The reason is that we just need to verify
conditions \eqref{eq:inexactConditions}. However, by Lemmas
\ref{lem:estimate1} and \ref{lem:estimate4}, they jointly hold with
probability at least $1 - 3/n$ provided that $m \geq \mu \cdot s \cdot
(19 \log n + 2)$ (recall that $\mu$ is a shorthand for $\mu(F)$).

\subsection{Proof of Lemma \ref{lem:dualVector}}  

The proof uses the clever {\em golfing scheme} introduced in
\cite{gross09}.  Partition $A$ into row blocks so that from now on,
$A_1$ are the first $m_1$ rows of the matrix $A$, $A_2$ the next $m_2$
rows, and so on. The $\ell$ matrices $\{A_i\}_{i = 1}^\ell$ are
independently distributed, and we have $m_1 + m_2 + \ldots + m_\ell =
m$. As before, $A_{i,T}$ is the restriction of $A_i$ to the columns in
$T$.

The golfing scheme then starts with $v_0 = 0$, inductively defines
\[
v_i = \frac{m}{m_i} A_i^* A_{i,T} (\sgn{x_T} - v_{i-1, T}) + v_{i-1}
\]
for $i = 1, \ldots, \ell$, and sets $v = v_\ell$.  Clearly $v$ is in
the row space of $A$. To simplify notation, let $q_i = \sgn{x_T} -
v_{i,T}$, and observe the two identities
\begin{equation} \label{eq:vT} q_i = \left(\Id - \frac{m}{m_i} A_{i, T}^*
    A_{i,T}\right)q_{i-1} = \prod_{j=1}^i \left(\Id - \frac{m}{m_j}
    A_{j,T}^* A_{j,T}\right)\sgn{x_T}
\end{equation}
and 
\begin{equation} \label{eq:vTc} v = \sum_{i=1}^\ell \frac{m}{m_i}A_i^*
  A_{i,T}\, q_{i-1}, 
\end{equation}
which shall be used frequently. From \eqref{eq:vT} and the fact that
$I - \frac{m}{m_i} A_{i, T}^* A_{i,T}$ should be a contraction
(local isometry {\bf E1}), we see that the norm of $q_i$ decreases
geometrically fast --- the terminology comes from this fact since each
iteration brings us closer to the target just as each golf shot would
bring us closer to the hole --- so that $v_T$ should be close to
$\sgn{x_T}$.  Hopefully, the process keeps the size of $v_{T^c}$ under
control as well.

To control the size of $v_{T^c}$ and that of $\sgn{x_T} - v_T$, we
claim that the following inequalities hold for each $i$ with high
probability: first,
\begin{equation}
\label{eq:c_i}
\twonorm{q_i} \leq c_i \twonorm{q_{i-1}}
\end{equation}
and, second,
\begin{equation}
\label{eq:t_i}
\infnorm{\frac{m}{m_i}A_{i,T^c}^* A_{i,T} \, q_{i-1}} \leq t_i \twonorm{q_{i-1}}
\end{equation}
(the values of the parameters $t_i$ and $c_i$ will be specified
later).  Let $p_1(i)$ (resp.~$p_2(i)$) be the probability that the
bound \eqref{eq:c_i} (resp.~\eqref{eq:t_i}) does not hold.  Lemma
\ref{lem:estimate2} gives
\begin{equation}
\label{eq:p_1}
p_1(i) \leq \exp\left(-\frac{1}{4} (c_i \sqrt{m_i/(s \mu)} - 1)^2\right).
\end{equation}
Thus, if
\begin{equation}
\label{eq:m_1}
m_i \geq \frac{2 + 8 (\beta + \log \alpha)}{c^2_i} s \mu,
\end{equation}
then $p_1(i) \leq \frac{1}{\alpha} e^{-\beta}$. Next, Lemma
\ref{lem:estimate3} gives
\begin{equation}
\label{eq:p_2}
p_2(i) \leq 2 n \exp\left(-\frac{3 t_i^2 m_i}{6 \mu + 2 \mu \sqrt{s} t_i}\right).
\end{equation}  
Thus, if
\begin{equation}
\label{eq:m_2}
m_i \geq \left(\frac{2}{t_i^2 s} + \frac{2}{3 t_i \sqrt{s}}\right)(\beta + \log(2 \alpha) + \log n) s \mu,  
\end{equation}
then $p_2(i) \leq \frac{1}{\alpha} e^{-\beta}$.

It is now time to set the number of blocks $\ell$, the block sizes
$m_i$ and the values of the parameters $c_i$ and $t_i$. These are as
follows:
\begin{itemize}
\item $\ell = \lceil (\log_2 s)/2 \rceil + 2$;
\item $c_1 = c_2 = 1/[2 \sqrt{\log n}]$ and $c_i = 1/2$ for $3 \le i
  \le \ell$;
\item $t_1 = t_2 = 1/[8 \sqrt{s}]$ and $t_i = \log n/[8 \sqrt{s}]$ for
  $3 \le i \le \ell$;
\item $m_1, m_2 \geq 35(1 + \log 4 + \beta) s \mu c_i^{-2}$ and $m_i
  \geq 35(1 + \log 6 + \beta) s \mu c_i^{-2}$ for $3 \le i \le \ell$.
\end{itemize} 
It is not hard to see that the total number of samples $m = \sum_i
m_i$ obeys the assumptions of the lemma. To see why $v$ is a valid
certificate, suppose first that for each $i$, \eqref{eq:c_i} and
\eqref{eq:t_i} hold.  Then \eqref{eq:vT} gives
\[
\twonorm{\sgn{x_T} - v_T} = \twonorm{q_\ell} \leq \twonorm{\sgn{x_T}}
  \, \prod_{i=1}^\ell c_i \leq \frac{\sqrt{s}}{2^\ell} \leq
  \frac{1}{4}
\] 
as desired.  Further, \eqref{eq:vTc} yields
\[
\infnorm{v_{T^c}} \leq \sum_{i=1}^\ell \infnorm{\frac{m}{m_i} A^*_{i,T^c}
  A_{i, T} q_{i-1}} \leq \sum_{i=1}^\ell t_i \twonorm{q_{i-1}} \leq
\sqrt{s} \, \sum_{i=1}^\ell t_i  \prod_{j=1}^{i-1} c_i.
\]
Now with our choice of parameters, the right-hand side is bounded
above by 
\[
\frac{1}{8}\left(1 + \frac{1}{2 \sqrt{\log n}} + \frac{\log n}{
    4 \log n} + \hdots\right) < \frac{1}{4},
\]
which is the desired conclusion.

Now we must show that the bounds \eqref{eq:c_i}, \eqref{eq:t_i} hold
with probability at least $1 - e^{-\beta} - 1/n$.  It follows from
\eqref{eq:m_1} and \eqref{eq:m_2} that $p_1(i), p_2(i) \leq
\frac{1}{4} e^{-\beta}$ for $i = 1,2$ and $p_1(i), p_2(i) \leq
\frac{1}{6} e^{-\beta} \leq 1/6$ for $i \geq 3$.  Thus, $p_1(1) +
p_1(2) + p_2(1) + p_2(2) \leq e^{-\beta}$ and $p_1(i) + p_2(i) \leq
1/3$ for $i \geq 3$.  Now the union bound would never show that
\eqref{eq:c_i} and \eqref{eq:t_i} hold with probability at least $1 -
1/n$ for all $i \geq 3$ because of the weak bound on $p_1(i) +
p_2(i)$.  However, using a clever idea in \cite{gross09}, it is not
necessary for each subset of rows to `succeed' and give the desired
bounds.  Instead, one can sample a `few' extra batches of rows, and
throw out those that fail our requirements.  We only need $\ell-2$
working batches, after the first 2.  In particular, pick $\ell' + 2 >
\ell$ batches of rows, so that we require $m \geq 2\cdot \lceil 140(1
+ \log 4 + \beta) \cdot \mu \cdot s \cdot \log n \rceil + \ell \cdot
\lceil 140 (1 + \log 6 + \beta) s \mu\rceil$ (note that we have made
no attempt to optimize constants).  Now as in \cite{gross09}, let $N$
be the the number of batches --- after the first 2 --- obeying
\eqref{eq:c_i} and \eqref{eq:t_i}; this $N$ is larger
(probabilistically) than a binomial($\ell', 2/3$) random variable.
Then a standard concentration bound \cite[Theorem
2.3a]{mcdiarmid}
\[
\P(N < \ell-2) \leq \exp\left(-2\, \frac{(\frac{2}{3} \ell' - \ell +
    2)^2}{\ell'}\right)
\] 
tells us that if we were to pick $\ell' = 3\lceil \log n \rceil + 1$,
we would have
\[
\P(N < \ell-2) \leq 1/n.
\]

In summary, from $p_1(1) + p_2(1) + p_1(2) + p_2 (2) \leq e^{- \beta}$
and the calculation above, the dual certificate $v$ obeys the required
properties with probability at least $1 - 1/n - e^{-\beta}$, provided
that $m \geq C (1 + \beta) \cdot \mu \cdot s \cdot \log n$.

\section{General Signal Recovery from Noisy Data}
\label{sec:noisy}

We prove the general recovery theorems from Section \ref{sec:main-result}
under the assumption of Gaussian white noise but would like to
emphasize that the same result would hold for other noise
distributions. Specifically, suppose we have the noisy model
\begin{equation}
\label{eq:noiseCorrelations}
y = Ax + z, \quad \text{where} \quad \infnorm{A^* z} \leq \lambda_n
\end{equation}
holds with high probability. Then the conclusions of Theorem
\ref{teo:noisyDS} remain valid. In details, the Dantzig selector with
constraint $\infnorm{A^*(y - A \bar x)} \le 4 \lambda_n$ obeys
\begin{equation}
\label{eq:twoBound2}
\twonorm{\hat{x} - x}  \leq C_1 (1 + \alpha^2) \,  \left[\frac{\onenorm{x - x_s}}{\sqrt{s}} + \lambda_n \sqrt{s}\right]
\end{equation}
with high probability. Hence, \eqref{eq:twoBoundDS} is a special case
corresponding to $\lambda_n = 2.5 \sigma_m \sqrt{\log n} = 2.5 \sigma
\sqrt{\frac{\log n}{m}}$.  Likewise, the bound on the $\ell_1$ loss
\eqref{eq:oneBoundDS} with $\lambda_n$ in place of $\sigma
\sqrt{\frac{\log n}{m}}$ holds as well. A similar generality applies
to the LASSO as well, although in this case we need a second noise
correlation bound, namely,
\[
\infnorm{A^*_{T^c} (I - P) z} \leq \lambda_n. 
\]

Now when $z \sim \mathcal{N}(0, \Id)$ and $A$ is a fixed matrix, we
have
\begin{equation}
\label{eq:boundnoise}
\infnorm{A^* z} \leq 2 \|A\|_{1,2} \sqrt{\log n}
\end{equation}
with probability at least $1-1/2n$; here, $\|A\|_{1,2}$ is the
maximum column norm of $A$.  Indeed, the $i$th component of $A^*z$ is
distributed as $\mathcal{N}(0, \twonorm{A_{\{i\}}}^2)$ and, therefore,
the union bound gives
\[
\P(\infnorm{A^* z} > 2 \|A\|_{1,2} \sqrt{\log n}) \le n
\P(|\mathcal{N}(0,1)| > 2 \sqrt{\log n}).
\]
The conclusion follows for $n \ge 2$ from the well-known tail bound
$\P(|\mathcal{N}(0,1)| > t) \le 2\phi(t)/t$, where $\phi$ is the
density of the standard normal distribution. The same steps
demonstrate that
\begin{equation}
\label{eq:noiseCorrelation2}
\infnorm{A^* (I - P)z} \leq 2 \|(I-P) A\|_{1,2} \sqrt{\log n} \leq 2 \|A\|_{1,2} \sqrt{\log n}
\end{equation}
with probability at least $1 - 1/2n$.

\subsection{Proof of Theorem \ref{teo:noisy}}
\label{sec:lasso-proof}

We begin with a few simplifying assumptions. First, we assume in the
proof that $\sigma_m = 1$ since the general result follows from a simple
rescaling. Second, because we are interested in situations where $m$
is much smaller than $n$, we assume for simplicity of presentation
that $m \leq n$ although our results extend with only a change to the
numerical constants involved if $m \leq n^{O(1)}$. In truth, they
extend without any assumption on the relation between $m$ and $n$, but
the general presentation becomes a bit more complicated.

Fix $s$ obeying $s \leq \bar s$, and let $T = \supp{x_s}$. We prove
the error bounds of Theorem \ref{teo:noisy} with $s$ fixed, and the
final result follows by considering that $s$ which minimizes either
the $\ell_2$ \eqref{eq:twoBound} or $\ell_1$ \eqref{eq:oneBound} error
bound. This is proper since the minimizing $s$ has a deterministic
value.  With $T$ as above, we assume in the rest of the proof that
\begin{itemize}
\item[(i)] all of the requirements for noiseless recovery in Lemma
  \ref{lem:inexactDuality} are met,
\item[(ii)] and that the inexact dual vector $v$ of Section
  \ref{sec:noiseless} is successfully constructed.
\end{itemize}
All of this occurs with probability at least $1 - 4/n - e^{-\beta}$.
Further, we assume that 
\begin{itemize}
\item[(iii)] the weak RIP holds with $\delta = 1/4$, $r = \frac{m}{C
    (1 + \beta) \cdot \mu(F) \cdot \log^5 n} \vee s$ and $T$ is as
  above. 
\end{itemize}
This occurs with probability at least $1 - 5 e^{-\beta}$, and implies
the RIP at sparsity level $r$ and restricted isometry constant $\delta
= 1/4$.  Lastly, we assume 
\begin{itemize}
\item[(iv)] the noise correlation bound 
\begin{equation}
\label{eq:precisebound}
\infnorm{A^* z} \leq 2.5 \sqrt{\log n}. 
\end{equation}
\end{itemize}
Assuming the weak RIP above, which implies $\|A\|_{1,2} \le 5/4$, the
conditional probability that this occurs is at least $1-1/2n$ because
of \eqref{eq:boundnoise}. Because the weak RIP implies the local isometry
condition {\bf E1} with $\delta = 1/4$, all of these conditions
together hold with probability at least $1 - 4/n - 6e^{-\beta}$.  All
of the steps in the proof are now deterministic consequences of
(i)--(iv); from now on, we will assume they hold.

With $h = \hat x - x$, our goal is to bound both the $\ell_2$ and
$\ell_1$ norms of $h$.  We will do this with a pair of lemmas.  The
first is frequently used (recall that $\lambda$ is set to
$10\sqrt{\log n}$).
\begin{lemma}[Tube constraint]
\label{lem:cone}
The error $h$ obeys
\[
\infnorm{A^* A h} \leq \frac{5 \lambda}{4}.
\]
\end{lemma}
\begin{proof}
  As shown in \cite[Lemma 3.1]{CP07}, writing that the zero vector is
  a subgradient of the LASSO functional $\frac12 \|y - A\bar
  x\|_{\ell_2}^2 + \lambda \|\bar x\|_{\ell_1}$ at $\bar x = \hat x$
  gives
 \[
 \infnorm{A^*(y - A \hat{x})} \leq \lambda.
 \]
 Then it follows from the triangle inequality that
 \[
 \infnorm{A^*A h} \leq \infnorm{A^* (y - A\hat{x})} + \infnorm{A^* z} \leq \lambda + \infnorm{A^* z}, 
 \]
 where $z$ is our noise term.  The claim is a consequence of
 \eqref{eq:precisebound}.
\end{proof}

\begin{lemma}
\label{lem:hTc}
The error $h$ obeys
\begin{equation}
\label{eq:hTc}
\onenorm{h_{T^c}} \leq C_0 (s \lambda + \onenorm{x_{T^c}}) 
\end{equation}
for some numerical constant $C_0$. 
\end{lemma}
Before proving this lemma, we show that it gives Theorem
\ref{teo:noisy}. Some of the steps are taken from the proof of
Theorem 1.1 in \cite{DS}.

\newcommand{\Tone}{{\bar{T}_1}}
\begin{proof}[Theorem \ref{teo:noisy}]
  Set $r$ as in (iii) above. We begin by partitioning $T^c$ and let
  $T_1$ be the indices of the $r$ largest entries of $h_{T^c}$, $T_2$
  be those of the next $r$ largest, and so on. We first bound
  $\twonorm{h_{T \cup T_1}}$ and set $\Tone = T \cup T_1$ for short.
  The weak RIP assumption (iii) gives
  \begin{equation}
    \frac{3}{4} \twonorm{h_{\Tone}}^2  \leq \twonorm{A_{\Tone} h_{\Tone}}^2 = \<A_{\Tone} h_{\Tone}, Ah\> - \<A_{\Tone} h_{\Tone}, A_{\Tone^c} h_{\Tone^c}\>. \label{eq:hTbound}
  \end{equation}
  From Lemma \ref{lem:cone}, we have
\[
\<A_{\Tone} h_{\Tone}, Ah\> = \<h_{\Tone}, A_{\Tone}^* A h\> \leq
\onenorm{h_{\Tone}} \, \infnorm{A_{\Tone}^* A h} \le \frac{5}{4}
\lambda \onenorm{h_{\Tone}}.
\]
Since $\Tone$ has cardinality at most $2s$, the Cauchy-Schwartz
inequality gives
\begin{equation}
\label{eq:firstTerm}
\<A_{\Tone} h_{\Tone}, Ah\> \leq 
\frac{5}{4} \lambda \sqrt{2s}  \twonorm{h_{\Tone}}. 
\end{equation}
Next, we bound $\abs{\<A_{\Tone}h_{\Tone}, A_{\Tone^c} h_{\Tone^c}\>}
\leq \abs{\<A_T h_{T}, A_{\Tone^c} h_{\Tone^c}\>} +
\abs{\<A_{T_1}h_{T_1}, A_{\Tone^c} h_{\Tone^c}\>}$.  We have
\begin{equation}
\label{eq:innerProduct}
\<A_T h_T, A_{\Tone^c} h_{\Tone^c}\> \leq \sum_{j\geq 2} \abs{\<A_T h_T, A_{T_j} h_{T_j}\>}.
\end{equation}
As shown in \cite[Lemma 1.2]{CT05}, the parallelogram identity
together with the weak RIP imply that
\[
\abs{\<A_T h_T, A_{T_j} h_{T_j}\>} \leq \frac{1}{4} \twonorm{h_T}
\twonorm{h_{T_j}}
\] 
and, therefore, 
\begin{equation}
\label{eq:innerProductT}
\<A_T h_T, A_{\Tone^c} h_{\Tone^c}\> \leq 
\frac{1}{4} \twonorm{h_T} \sum_{j\geq 2} \twonorm{h_{T_j}}.
\end{equation}
To bound the summation, we use the now standard result
\cite[(3.10)]{DS} 
\begin{equation}
\label{eq:summation}
\sum_{j\geq 2} \twonorm{h_{T_j}} \leq r^{-1/2} \onenorm{h_{T^c}},  
\end{equation}
which gives 
\[
\abs{\<A_T h_T, A_{\Tone^c} h_{\Tone^c}\>} \leq \frac{1}{4} r^{-1/2}
\twonorm{h_T} \onenorm{h_{T^c}}.
\] 
The same analysis yields $\abs{\<A_{T_1} h_{T_1}, A_{\Tone^c}
  h_{\Tone^c}\>} \leq \tfrac{1}{4} r^{-1/2} \twonorm{h_{T_1}} \onenorm{h_{T^c}}$
and thus,
\[
\abs{\<A_{\Tone} h_{\Tone}, A_{\Tone^c} h_{\Tone^c}\>} \leq \frac12 r^{-1/2}
\twonorm{h_{\Tone}} \onenorm{h_{T^c}}.
\]
Plugging these estimates into \eqref{eq:hTbound} gives
\begin{equation}
  \twonorm{h_{\Tone}} \leq \frac12 \Big(\frac{5}{2}\sqrt{2s} \lambda + r^{-1/2} \onenorm{h_{T^c}}\Bigr).\label{eq:hTBounded}
\end{equation}

The conclusion is now one step away.  Obviously,
\begin{align*}
  \twonorm{h}  \leq \twonorm{h_{\Tone}} + \sum_{j\geq 2}
  \twonorm{h_{T_j}} 
  & \leq \twonorm{h_{\Tone}} + r^{-1/2} \onenorm{h_{T^c}}\\
  & \le \frac12 \Big(\frac{5}{2}\sqrt{2s} \lambda + 3r^{-1/2} \onenorm{h_{T^c}}\Bigr),
\end{align*}
where the second line follows from \eqref{eq:hTBounded}.  Lemma
\ref{lem:hTc} completes the proof for the $\ell_2$ error.  For the
$\ell_1$ error, note that by the Cauchy-Schwartz
inequality
\[\onenorm{h} = \onenorm{h_T} + \onenorm{h_{T^c}} \leq \sqrt{s} \twonorm{h_T}  + \onenorm{h_{T^c}} \leq \sqrt{s} \twonorm{h_{\Tone}} + \onenorm{h_{T^c}}.\]
Combine this with \eqref{eq:hTBounded} and Lemma \ref{lem:hTc}.
\end{proof}

\subsection{Proof of Lemma \ref{lem:hTc}}

Since $\hat{x}$ is the minimizer to \eqref{eq:Lasso},
\[
\frac12 \twonorm{A \hat{x} - y}^2 + \lambda \onenorm{\hat{x}} \leq
\frac12 \twonorm{A x - y}^2 + \lambda \onenorm{x},
\]
which can be massaged into the more convenient form
\begin{equation*}
\frac12 \twonorm{A h}^2 + \lambda \onenorm{\hat{x}} \leq \<Ah, z\> + 
\lambda \onenorm{x}.
\end{equation*}
\begin{lemma}
\label{teo:cone2} 
\[
\onenorm{\hat{x}} \geq \onenorm{x} + \<h_T, \text{\em sgn}(x_T)\> +
\onenorm{h_{T^c}} - 2\onenorm{x_{T^c}}.
\]
\end{lemma}
\begin{proof} 
  We have $\onenorm{\hat{x}} = \<\hat{x}, \sgn{\hat{x}}\> \geq \<x_T +
  h_T, \sgn{x_T}\> + \onenorm{x_{T^c} + h_{T^c}}$ and the claim
  follows from the triangle inequality.
\end{proof} 
It follows from this that
\begin{equation}
\label{eq:first}
\frac{1}{2} \twonorm{A h}^2 + \lambda \onenorm{h_{T^c}} \leq \<Ah, z\> -
\lambda \<h_T, \sgn{x_T}\> + 2\lambda \onenorm{x_{T^c}},  
\end{equation}
and the proof is now a consequence of the two short lemmas below.

\begin{lemma}
  \label{lem:Ahz} 
\begin{equation}
\label{eq:AhzBounded}
\<Ah, z\> \leq \frac{5}{12} s \lambda^2 + \frac{\lambda}{4} \onenorm{h_{T^c}}.
\end{equation}
\end{lemma}
\begin{proof}
  The proof is similar to an argument in \cite{CP07}. Let $P = A_T
  (A^*_T A_T)^{-1} A^*_T$ be the orthogonal projection onto the range
  of $A_T$.  Then
\begin{align}
\<Ah, z\> &= \<P A h, z\> + \<(\Id - P) A_{T^c} h_{T^c}, z\>\nonumber\\
 &= \<A^*_T A h, (A^*_T A_T)^{-1} A^*_T z\> + \<h_{T^c}, A^*_{T^c} (\Id - P)z\>\nonumber\\
 &\leq \infnorm{A_T^* A h} \onenorm{(A^*_T A_T)^{-1} A^*_T z} + \onenorm{h_{T^c}} \infnorm{A^*_{T^c} (\Id - P)z}\nonumber\\
 &\leq \frac{5}{4} \lambda \onenorm{(A^*_T A_T)^{-1} A^*_T z} + 2.5 \sqrt{\log n} \onenorm{h_{T^c}}.\label{eq:boundAhz}
 \end{align}
 The last line follows from Lemma \ref{lem:cone} and
 \eqref{eq:noiseCorrelation2}. 
  We now
 bound the first term, and write
\begin{align}
  \onenorm{(A^*_T A_T)^{-1} A^*_T z} &\leq \sqrt{s} \twonorm{(A^*_T A_T)^{-1} A^*_T z}\nonumber\\
  &\leq \frac{4}{3} \sqrt{s} \twonorm{A^*_T z}\nonumber\\
  &\leq \frac{4}{3} s \infnorm{A^*_T z} \leq \frac{1}{3} s\, 
  \lambda. \label{eq:AAAz}
\end{align}
The first inequality follows from Cauchy-Schwartz, the second from
$\|A_T^* A_T\| \le 4/3$, and the fourth from $\|A^* z\|_{\ell_\infty}
\le \lambda/4$. Inequality \eqref{eq:boundAhz} establishes the claim.
\end{proof}

\begin{lemma}
\label{lem:hTsgnxT} 
\begin{equation}
  \abs{\<h_T, \text{\em sgn}(x_T)\>} \leq C s \lambda + \frac{7}{12} \onenorm{h_{T^c}} + \frac{1}{2 \lambda} \twonorm{A h}^2. 
\end{equation}
\end{lemma}
\begin{proof}
  Let $v$ be the inexact dual vector, and decompose $\<h_T,
  \sgn{x_T}\>$ as
\begin{align}
  |\<h_T, \sgn{x_T}\>| &
  \leq |\<h_T, \sgn{x_T} - v_T\>| + |\<h_T, v_T\>| \nonumber\\
  &\leq |\<h_T, \sgn{x_T} - v_T\>| + |\<h, v\>| + |\<h_{T^c},
  v_{T^c}\>|. \label{eq:boundhsgn}
\end{align}
First, 
\[
|\<h_T,
\sgn{x_T} - v_T\>| \leq \twonorm{h_T} \twonorm{\sgn{x_T} - v_T} \leq
\tfrac{1}{4} \twonorm{h_T}.
\]  
Now
\begin{align}
  \twonorm{h_T} \leq \opnorm{{(A_T^* A_T)^{-1}}}\twonorm{A_T^* A_T
    h_T}
  &\leq \frac{4}{3} \twonorm{A_T^* A_T h_T}\nonumber\\
  &\leq \frac{4}{3} \twonorm{A_T^* A h} + \frac{4}{3} \twonorm{A_T^* A_{T^c} h_{T^c}}\nonumber\\
  &\leq \frac{4}{3} \sqrt{s} \infnorm{A_T^* A h} + \frac{4}{3} \onenorm{h_{T^c}}\max_{j\in T^c} \twonorm{A_T^* A_{\{j\}}}\nonumber\\
  &\leq \frac{5}{3} \sqrt{s} \lambda + \frac{4}{3}
  \onenorm{h_{T^c}}, \label{eq:boundhT}
\end{align}
where the last line follows from Lemma \ref{lem:cone} and
\eqref{eq:inexactConditions}. Second, it follows from the definition
of $v$ that
\[
\abs{\<h_{T^c}, v_{T^c}\>} \leq \onenorm{h_{T^c}} \infnorm{v_{T^c}}
\leq \tfrac{1}{4} \onenorm{h_{T^c}}. 
\]
Hence, we established 
\begin{equation}
\label{eq:boundhTxTlate}
\abs{\<h_T, \sgn{x_T}\>} \leq \frac{5}{12} \sqrt{s} \lambda + \frac{7}{12} \onenorm{h_{T^c}} + \abs{\<h, v\>}. 
\end{equation}
Third, we bound $\abs{\<h, v\>}$ by Lemma \ref{lem:w} below. With the
notation of this lemma, 
\[
\abs{\<h, v\>} = \abs{\<h, A^* w\>} = \abs{\<Ah, w\>} \le \twonorm{Ah}
\twonorm {w} \le C_0 \sqrt{s} \twonorm{Ah}
\]
for some $C_0 > 0$.  Since
\[
\twonorm{Ah} \sqrt{s} \leq \frac{\twonorm{Ah}^2}{2 C_0 \lambda} +
\frac{C_0 s \lambda}{2}, 
\]
it follows that 
\begin{equation}
\label{eq:hvBounded}
\abs{\<h, v\>} \leq \frac{C_0^2}{2} s \lambda + \frac{1}{2 \lambda} \twonorm{A h}^2.  
\end{equation}
Plugging this into \eqref{eq:boundhTxTlate} finishes the proof.
\end{proof}

\begin{lemma}
\label{lem:w} 
The inexact dual certificate from Section \ref{sec:noiseless} is of
the form $v = A^* w$ where $\twonorm{w} \le C_0 \sqrt{s}$ for some
positive numerical constant $C_0$.
\end{lemma}
\begin{proof} For notational simplicity, assume without loss of
  generality that the first $\ell$ batches of rows were those used in
  constructing the dual vector $v$ (none were thrown out) so that
\[
v = \sum_{i=1}^\ell \frac{m}{m_i} A^*_i A_{i, T} q_{i-1}. 
\] 
Hence, $v = A^* w$ with $w^* = (w_1^*, \ldots, w_\ell^*, 0, \ldots,
0)$ and $w_i := \frac{m}{m_i} A_{i, T} q_{i-1}$ so that $\twonorm{w}^2
= \sum_{i = 1}^\ell \twonorm{w_i}^2$.  We have
\begin{align}
\tfrac{m}{m_i} \twonorm{A_{i,T}\, q_{i-1}}^2 &= \<\tfrac{m}{m_i} A_{i,T}^*A_{i,T} q_{i-1}, q_{i-1}\>\nonumber\\
&= \<(\tfrac{m}{m_i} A_{i,T}^*A_{i,T} - \Id)q_{i-1}, q_{i-1}\> + \twonorm{q_{i-1}}^2\nonumber\\
&\leq \twonorm{q_i} \twonorm{q_{i-1}} + \twonorm{q_{i-1}}^2\nonumber\\
& \leq 2 \twonorm{q_{i-1}}^2\nonumber\\
& \leq 2s \prod_{j=1}^{i-1} c_j^2.
\end{align}
It follows that 
\[
\twonorm{w}^2 \leq 2s \cdot \sum_{i=1}^\ell \frac{m}{m_i}
\prod_{j=1}^{i-1} c_j^2.
\] 
Assume that $m \leq C (1 + \beta) \mu s \log n$ so that $m$ is just
large enough to satisfy the requirements of Theorem \ref{teo:noisy}
(up to a constant).  Then recall that $m_i \geq C(1+ \beta) \mu s
c_i^{-2} \Rightarrow \tfrac{m}{m_i} \leq C c_i^2 \log n $.  (If $m$ is
much larger, rescale each $m_i$ proportionally to achieve the same
ratio.) This gives
\[
\twonorm{w}^2 \leq C s \log n \, \sum_{i=1}^\ell \prod_{j=1}^{i} c_j^2
\le C s \, \sum_{i=1}^\ell \prod_{j=2}^{i} c_j^2.
\]
since $c_1 = (2\sqrt{\log n})^{-1}$. For $i \ge 1$, $\prod_{j=2}^i
4^{-(i-1)}$ and the conclusion follows. 
\end{proof}

\subsection{Proof of Theorem \ref{teo:noisyDS}}



\begin{proof}
  Fix $s$ and $T$ as in Section \ref{sec:lasso-proof} and assume that
  (i)--(iv) hold.  The proof parallels that for the LASSO; this
  is why we only sketch the important points and reuse the earlier
  techniques with minimal extra explanation.  We shall repeatedly use
  the inequality
\begin{equation}
\label{eq:DS-useful-ineq}
ab \leq c a^2/2 + b^2/(2c), 
\end{equation}
which holds for positive scalars $a,b,c$. Our first intermediate
result is analogous to Lemma \ref{lem:hTc}. 
\begin{lemma}
\label{lem:DS-lemma}
The error $h = \hat x - x$ obeys
\[
\onenorm{h_{T^c}} \leq 
C (s \lambda + \onenorm{x_{T^c}} + \sqrt{s} \twonorm{A h}).
\]
\end{lemma}
\begin{proof}
Since $x$ is feasible, $\onenorm{\hat{x}} \leq \onenorm{x}$ and it follows from Lemma \ref{teo:cone2} that 
\begin{equation} \label{eq:DS-tube} \onenorm{h_{T^c}} \leq -\<h_T,
  \sgn{x_T}\> + 2 \onenorm{x_{T^c}}.
\end{equation}
We bound $|\<h_T, \sgn{x_T}\>|$ in exactly the same way as before, but
omitting the last step, and obtain 
\[
|\<h_T, \sgn{x_T}\>| \leq C s \lambda + \frac{7}{12} \onenorm{h_{T^c}}
+ C \sqrt{s} \twonorm{A h}.
\]
This concludes the proof.
\end{proof}

The remainder of this section proves Theorem
\ref{teo:noisyDS}. Observe that $\|A^* Ah \|_{\ell_\infty} \le
\tfrac{5}{4} \lambda$ (Lemma \ref{lem:cone}) since the proof is
identical (we do not even need to consider subgradients).
Partitioning the indices as before, one can repeat the earlier
argument leading to \eqref{eq:hTBounded}. Then combining
\eqref{eq:hTBounded} with Lemma \ref{lem:DS-lemma} gives
\begin{equation}
\label{eq:DS-many-terms}
\twonorm{h_{\Tone}} \leq C\sqrt{s} \lambda + 
C r^{-1/2}(s \lambda + \onenorm{x_{T^c}} + \sqrt{s} \twonorm{A h}).
\end{equation}
The term proportional to $\sqrt{s/r} \twonorm{A h}$ in the right-hand
side was not present before, and we must develop an upper bound for
it.  Write
\[
\twonorm{A h}^2 = 
\<A^* A h, h\> \leq \infnorm{A^* A h} \onenorm{h} 
\leq \frac{5}{4} \lambda (\onenorm{h_T} + \onenorm{h_{T^c}}) 
\]
and note that \eqref{eq:DS-tube} gives
\[
\onenorm{h_{T^c}} \leq
\onenorm{h_T} + 2 \onenorm{x_{T^c}}.
\] 
These last two inequalities yield $\twonorm{A h}^2 \le \tfrac{5}{2}
\lambda (\onenorm{h_T} + \onenorm{x_{T^c}})$, and since $\sqrt{\lambda
  \onenorm{ x_{T^c}}} \le \frac12 \lambda \sqrt{s} +
\frac{1}{2\sqrt{s}} \onenorm{x_{T^c}}$ because of
\eqref{eq:DS-useful-ineq}, we have
\[
\twonorm{A h} \leq \sqrt{\tfrac{5}{2} \lambda} (\sqrt{\onenorm{h_T}} +
\sqrt{\onenorm{x_{T^c}}}) \le \sqrt{\tfrac{5}{2}}\Bigl(
\sqrt{\lambda \onenorm{h_T}} + \tfrac{1}{2} \lambda \sqrt{s} +
  \tfrac{1}{2\sqrt{s}} \onenorm{x_{T^c}}\Bigr).
\]
In short, 
\[
\twonorm{h_{\Tone}} \leq C \Bigl(\sqrt{s} \lambda +  r^{-1/2}\Bigl(s
\lambda + \onenorm{x_{T^c}} + \sqrt{s \lambda \onenorm{h_T}}\Bigr)\Bigr).
\] 
The extra term on the right-hand side has been transmuted into $C
\sqrt{\tfrac{s}{r} \lambda \onenorm{h_T}}$,  which  may
be bounded via 
\eqref{eq:DS-useful-ineq} as 
\[ 
C \sqrt{\frac{s}{r} \lambda \onenorm{h_T}} \leq C^2
\frac{s}{r} \, \sqrt{s} \lambda + \frac{1}{2\sqrt{s}}
  \onenorm{h_{T}} \le C^2 \frac{s}{r} \, \sqrt{s} \lambda +
  \frac{1}{2} \twonorm{h_T}.
\]
Since $\twonorm{h_T} \leq \twonorm{h_{\Tone}}$, we have 
\[
\twonorm{h_\Tone} \leq C \, \Bigl(1 + \sqrt{\frac{s}{r}} + \frac{s}{r}\Bigr)
\, \sqrt{s} \lambda + C \, \frac{\onenorm{x_{T^c}}}{\sqrt{r}}. 
\]
The remaining steps are the same as those in the proof for the LASSO.
\end{proof}

\section{Discussion}
\label{sec:discussion}

This paper developed a very simple and general theory of compressive
sensing, in which sensing vectors are drawn independently at random
from a probability distribution. In addition to establishing a general
framework, we showed that nearly sparse signals could be accurately
recovered from a small number of noisy compressive samples by means of
tractable convex optimization. For example, $s$-sparse signals can be
recovered accurately from about $s \log n$ DFT coefficients corrupted
by noise. Our analysis shows that stable recovery is possible from a
minimal number of samples, and improves on previously known results.
This improvement comes from novel stability arguments, which do not
require the restricted isometry property to hold.

We have seen that the isotropy condition is not really necessary, and
it would be interesting to know the extent in which it can be
relaxed. In particular, for which values of $\alpha$ and $\beta$
obeying $\alpha \Id \preceq \E a a^* \preceq \beta \Id$ would our
results continue to hold? Also, we have assumed that the sensing
vectors are sampled independently at random, and although the main
idea in compressive sensing is to use randomness as a sensing
mechanism, it would be interesting to know how the results would
change if one were to introduce some correlations.

\appendix
\section{Proof of Theorem \ref{teo:weakRIP} (the weak RIP)}

Our proof uses some the results and techniques of \cite{rudelson99}
and \cite{rudelsonVershynin}. Recall that $A$ is a matrix with rows
drawn independently from a probability distribution $F$ obeying the
isotropy and incoherence conditions, and that we wish to show that for
any fixed $0 \leq \delta < 1$,
\[
(1 - \delta) \twonorm{v}^2 \leq \twonorm{A v}^2 \leq (1 + \delta)
\twonorm{v}^2. 
\] 
These inequalities should hold with high probability, uniformly over
all vectors $v$ obeying $\supp{v} \subset T \cup R$ where $T$ is
fixed, $R$ may vary, and
\[
|T| \leq c \frac{m}{\mu \log m}, \qquad |R| \leq c
\frac{m}{\mu \log n \log^4 m}.
\]
To express this in another way, set 
\[
X \coloneq \sup_{v \in V} \abs{\twonorm{A v}^2 - \twonorm{v}^2},
\] 
where 
\begin{equation}
\label{eq:V}
V = \{v : \twonorm{v} = 1, \, \supp{v} \subset T \cup R, \, |R| \le r,
\, T \cap R = \emptyset\}.
\end{equation}
In words, $v$ is a unit-normed vector supported on $T \cup R$ where
$T$ is fixed of cardinality $s \leq c m/(\mu \log m)$, and $R$ is any set disjoint from $T$
of cardinality at most $r \leq c m/(\mu\log n \log^4 m)$.  We wish to show that $X \leq \delta$ with
high probability.  We will first bound this random variable in
expectation and then show that it is unlikely to be much larger than
its expectation.  The bound in expectation is contained in the
following lemma.
\begin{lemma}
\label{lem:expectedDeviation}
Fix $\epsilon > 0$. Suppose $m \geq C \, \mu \, [s \log m \vee
r \log n \log^4 m]$, where $C$ is a constant only depending on
$\epsilon$. Then
\[
\E X \leq \epsilon.
\]
\end{lemma}
  
To begin the proof, note that for any $v$ with $\supp{v} \subset T
\cup R$, we have
\[
\twonorm{A v}^2 = \twonorm{A_T v_T}^2 + \twonorm{A_R v_R}^2 + 2\<v_T,
A_T^* A_R v_R\>.
\]
The first two terms are easily dealt with using prior results. To be
sure, under the conditions of Lemma \ref{lem:expectedDeviation}, a
slight modification of the proof of Theorem 3.4 in
\cite{rudelsonVershynin} gives\footnote{Rudelson and Vershynin
  consider a slightly different model but the proof in
  \cite{rudelsonVershynin} extends to our model with hardly any
  adjustments.}
\begin{equation}
\label{eq:boundR2}
\E \sup_{v_R : \abs{R}\leq r} \abs{\twonorm{A_R v_R}^2 - \twonorm{v_R}^2} \leq \frac{\epsilon}{4}  \twonorm{v_R}^2.
\end{equation}
Next, it follows from \cite{RudelsonIsotropic99}, or the matrix
Bernstein inequality in Estimate 1, that
\begin{equation}
\label{eq:boundR1}
\E \sup_{v_T} \abs{\twonorm{A_T v_T}^2 - \twonorm{v_T}^2} \leq \frac{\epsilon}{4} \twonorm{v_T}^2.
\end{equation}
Thus, to prove Lemma \ref{lem:expectedDeviation}, it suffices to prove
that
\[
\E \max_R \opnorm{A_R^* A_T} \leq \epsilon/4. 
\]
This is the content of the following theorem.

\begin{theorem} \label{teo:comboBound} Under the assumptions of Lemma
  \ref{lem:expectedDeviation}, we have
\begin{equation}
  \E \max_R \opnorm{A_R^* A_T} 
  \leq C \left(\sqrt{\frac{s \mu \log m}{m}} + \sqrt{\frac{r \mu \log n \log^3 m}{m}}\right). 
\end{equation}
\end{theorem}
Put differently, the theorem develops a bound on 
\begin{equation}
\label{eq:normATAR}
\E \max_{(x,y) \in B \times D}\, \frac{1}{m}
\sum_{i = 1}^m \<a_i, x\> \<a_i, y\>
\end{equation}
in which
\begin{align*}
  B &\coloneq \{x : \twonorm{x} \leq 1, \supp{x} \subset T\},\\
  D &\coloneq \{y : \twonorm{y} \leq 1, \supp{y} \cap T = \emptyset,\,
  \abs{\supp{y}} \leq r\}.
\end{align*}  
By symmetrization followed by a comparison principle -- both of which follow by Jensen's inequality (see \cite[Lemma
6.3]{ledouxTalagrand} followed by \cite[inequality
(4.8)]{ledouxTalagrand}), \eqref{eq:normATAR} is less or equal to a
numerical constant times
\[
\E \max_{\substack{(x,y) \in B \times D}} \frac{1}{m} \sum_{i = 1}^m
g_i \<a_i, x\> \<a_i, y\>, 
\] 
where the $g_i$'s are independent $\mathcal{N}(0,1)$ random variables.
The main estimate is a bound on the conditional expectation of the
right-hand side; that is, holding the vectors $a_i$ fixed.
\begin{lemma}[Main lemma]
\label{lem:comboBound}
Fix vectors $\{a_i\}_{i = 1}^m$ and let
\[
R_1 \coloneq \max_{x \in B} \frac{1}{m} \sum_{i = 1}^m \<a_i, x\>^2,
\qquad R_2  \coloneq \max_{y \in D} \frac{1}{m} \sum_{i = 1}^m
\<a_i, y\>^2.
\]
Suppose $m \geq C \, \mu \, [s \log m \vee r \log n \log^4 m]$.  Then
\[
\E \max_{\substack{(x,y) \in B \times D}} \, \frac{1}{m} \sum_{i =
  1}^m g_i \<a_i, x\> \<a_i, y\> \leq C \left(\sqrt{\frac{(1+R_2)
      \abs{T} \mu \log m }{m}} + \sqrt{\frac{(1 + R_1) s \mu \log n \log^3
      m}{m}}\right).
\]
\end{lemma}

\begin{proof}[Theorem \ref{teo:comboBound}] Under the assumptions of
  the theorem, it follows from the results in \cite{rudelsonVershynin}
  and Jensen's inequality that $\E \sqrt{1 + R_2}\leq \sqrt{1 + \E
    R_2} \leq C$.  Likewise, the results in \cite{RudelsonIsotropic99}
  and the same Jensen's inequality give $\E \sqrt{1 + R_1} \leq
  C$. (These inequalities were also noted, in a different form, in
  \eqref{eq:boundR1} and \eqref{eq:boundR2}).  Hence, Lemma
  \ref{lem:comboBound} implies
\[
\E \max_R \|A_R^* A_T\| \leq C \left(\sqrt{\frac{s \mu \log m }{m}} +
  \sqrt{\frac{r \mu \log n \log^3 m}{m}}\right).
\]
\end{proof}

\subsection{Proof of Lemma \ref{lem:comboBound}}

We need to develop a bound about the expected maximum of a Gaussian
process, namely, 
\[
\E \max_{\substack{(x,y) \in B \times D}} F(x,y),
\]
where
\[
F(x,y) \coloneq \sum_{i = 1}^m g_i \<a_i, x\> \<a_i, y\>.
\] 
We shall do this by means of the majorizing measure theorem below,
which may be found in \cite{rudelson99} and is attributed to Talagrand
(combine Theorem 4.1 with Propositions 2.3 and 4.4 in
\cite{talagrandMajorizing}). From now on, $(M,d)$ is a metric space and $B(t,
\epsilon)$ is the ball of center $t$ and radius $\epsilon$ under the
metric $d$.

\begin{theorem}[Majorizing measure theorem]
\label{teo:majorizing}
Let $(X_t)_{t\in M}$ be a collection of zero-mean random variables
obeying the subgaussian tail estimate
\begin{equation}
  \P(|X_t - X_{t'}| > u) \leq \exp\left(-c \frac{u^2}
    {d^2(t, t')}\right),
\end{equation}
for all $u > 0$.  Fix $\rho > 1$ and let $k_0$ be an integer so that
the diameter of $M$ is less than $\rho^{-k_0}$.  Suppose there exist
$\sigma > 0$ and a sequence of functions
$\{\varphi_k\}_{k=k_0}^\infty$, $\varphi_k: M \rightarrow \R^+$, with
the following two properties: 1) the sequence is uniformly bounded by
a constant depending only on $\rho$; 2) for each $k$ and for any $t
\in M$ and any points $t_1, \hdots, t_{\tilde{N}} \in B(t, \rho^{-k})$
with mutual distances at least $\rho^{-k-1}$, we have 
\begin{equation}
\label{eq:majreq}
\max_{j=1, \hdots, \tilde{N}} \varphi_{k+2}(t_j) \geq \varphi_k(s) + \sigma \rho^{-k} \sqrt{\log \tilde{N}}. 
\end{equation}
Then
\begin{equation}
\label{eq:majorExpected}
\E \sup_{t \in M} X_t \leq C(\rho) \cdot \sigma^{-1}.
\end{equation}
\end{theorem}

To apply this theorem, we begin by bounding the variance between
increments in order to ascertain the metric we need to use.  We
compute 
\begin{align*}
\Var(F(x,y) - F(x', y')) &= \sum_{i = 1}^m \Bigl(\<a_i, x\>\<a_i, y\> - \<a_i, x'\>\<a_i, y'\>\Bigr)^2\\
&= \sum_{i = 1}^m (\<a_i, x - x'\>\<a_i, y\> - \<a_i, x'\>\<a_i, y' - y\>)^2\\
&\leq 2 \sum_{i=1}^m \<a_i, x - x'\>^2 \<a_i, y\>^2 + 2 \sum_{i=1}^m \<a_i, x'\>^2 \<a_i, y - y'\>^2\\
&\leq 2 \opnorm{x - x'}_B^2 + 2 m R_1 \opnorm{y - y'}_D^2,  
\end{align*}
where we define the norms $\opnorm{\cdot}_B, \opnorm{\cdot}_D$ as follows:
\[
\opnorm{x}_B \coloneq \max_{y \in D} \sqrt{\sum_{i=1}^m \<a_i, x\>^2
  \<a_i, y\>^2}, \qquad \opnorm{y}_D \coloneq \max_{1 \leq i \leq m}
\abs{\<a_i, y\>}.
\] 
(We note that they may be pseudo norms, but this makes no difference
to the proof.  All of the utilized lemmas and theorems generalize to
pseudo norms.)  Thus, since $\sqrt{c + d} \leq \sqrt{c} + \sqrt{d}$
for any scalars $c,d$, we can use the metric
\[d((x,y), (x', y')) \coloneq \sqrt{2}  \opnorm{x - x'}_B + \sqrt{2 m R_1} \opnorm{y - y'}_D.\]

Before continuing, we record two useful lemmas for bounding
$\tilde{N}$. Here and below, $N(M, d, \epsilon)$ is the covering
number of $M$ in the metric $d$.
\begin{lemma}[Packing number bound]
\label{lem:packing}
Let $t_1, t_2, \hdots, t_{\tilde{N}} \in M$ be points with mutual
distances at least $2\epsilon$ under the metric $d$.  Then 
\[
\tilde{N} \leq N(M, d, \epsilon).
\]
\end{lemma}
This is a standard result proved by creating an injective mapping from
the points $\{t_j\}$ to those in the cover set (map each $t_j$ to the
nearest point in the cover).

The next lemma is a standard tool used to obtain bounds on covering
numbers, see \cite{pajor86} and \cite{bourgain89} for a more general
statement.
\begin{lemma}[Dual Sudakov minorization]
\label{lem:sudakov}
Let $B_{\ell_2}$ be the unit $\ell_2$ ball in $\R^d$, and let
$\opnorm{\cdot}$ be a norm.  Let $z \in \R^d$ be a Gaussian vector
with independent $\mathcal{N}(0,1)$ entries. Then there is a numerical
constant $C > 0$ such that
\[
\sqrt{\log N(B_{\ell_2}, \opnorm{\cdot}, \epsilon)} \leq
\frac{C}{\epsilon} \sqrt{\E \opnorm{z}^2}.
\]
\end{lemma}

We now invoke the majorizing measure theorem to prove Lemma
\ref{lem:comboBound}.  We start by bounding the diameter of $B
\times D$ under the metric $d$.  For any $x \in B$, $\abs{\<a_i,
  x\>} \leq \twonorm{a_{i, T}} \leq \sqrt{s \mu}$ and, likewise,
for any $y \in D$, $\abs{\<a_i, y\>} \leq \sqrt{r \mu}$. This gives
1) $\opnorm{x - x'}_B \leq 2 \sqrt{r s \mu^2 m}$ for any $x, x'
\in B$, 2) $\sqrt{m R_1} \leq \sqrt{s \mu m}$, and 3)
$\opnorm{y - y'}_D \leq 2 \sqrt{r \mu}$ for any $y, y' \in D$.
Combined, these bounds yield
\[
d((x,y), (x', y')) \leq 5\mu \, \sqrt{r s m}.
\]
Under the stated assumptions, the right-hand side is at most $m^{3/2}$
and we thus set $k_0$ to be the largest integer such that
\[
\rho^{-k_0} \geq m^{3/2}.
\]

We now define $\varphi_k$ on coarse and fine scales. In what follows,
we may take $\rho = 6$ so that $C(\rho)$ \eqref{eq:majorExpected} is
an absolute constant.
\begin{description}
\item[Coarse scales:] for $k = k_0, k_0 + 1, \ldots, 0$,
\[
\varphi_k(x) \coloneq \min\{\twonorm{u}^2 : \opnorm{u - x}_B \leq
\rho^{-k}\} + \frac{k - k_0}{\log n}.
\]
\item[Fine scales:] for $k \geq 1$, $\varphi_k$ is a constant function
  given by
\[
\varphi_k(x)  \coloneq 3 \rho \sigma \int_{\rho^{-k}}^1 \sqrt{\log N(B
  \times D, d, \epsilon)} d\epsilon + 3.
\]
\end{description}
Lastly, set
\[
\sigma^{-1} \coloneq C\sqrt{m} \, \left(\sqrt{(1 + R_2) s \mu \log m } +
  \sqrt{(1 + R_1) s \mu \log n \log^3 m}\right).
\]
Our definition of $\varphi_k$ is closely related to--and inspired by--the functions defined in \cite{rudelson99}. We need to show that these functions are uniformly bounded and obey
\eqref{eq:majreq} for all $k$.  We begin by verifying these properties
for fine scale elements as this is the less subtle calculation.

\subsection{Fine scale: $k \geq 1$}

To show that \eqref{eq:majreq} holds, observe that, 
\begin{align*}
  \varphi_{k+2} - \varphi_k & = 3 \sigma \rho \int_{\rho^{-(k + 2)}}^{\rho^{-k}} \sqrt{\log N(B \times D, d, \epsilon)} d\epsilon\\
  &\geq  3 \sigma \rho \int_{\rho^{-(k + 2)}}^{\frac12 \rho^{-(k+1)}} \sqrt{\log N(B \times D, d, \epsilon)} d\epsilon\\
    &\geq 3 \sigma \rho \Bigl(\frac12 \rho^{-(k+1)} - \rho^{-(k + 2)}\Bigr)\sqrt{\log N\Bigl(B \times D, d, \frac12 \rho^{-(k+1)}\Bigr)}\\
    &\geq \sigma \rho^{-k} \sqrt{\log \tilde{N}}.
\end{align*}
The last line follows from $\rho \ge 6$ and the packing number bound
(Lemma \ref{lem:packing}).  Note that
this same calculation holds when $k = 0, -1$ because for $k \leq 0$,
$\varphi_k \leq 3$ (see Section \ref{sec:coarseScale}).

We now show that $\varphi_k$ is bounded. Since  
\begin{equation}
\label{eq:fineInt}
\varphi_k \leq 3 \rho \sigma \int_0^1 \sqrt{\log N(B \times D, d, \epsilon)} d\epsilon + 3, 
\end{equation}
it suffices to show that the right-hand side is bounded.  This follows
from crude upper bounds on the covering number. Indeed, observe that
\begin{align*}
d((x,y), (x', y')) &\leq \sqrt{2 m R_2} \max_{1\leq i \leq m} \abs{\<a_i, x - x'\>} + \sqrt{2 m R_1} \max_{1\leq i \leq m} \abs{\<a_i, y - y'\>}\\
&\leq  \sqrt{2 m R_2 s \mu}  \twonorm{x - x'} + \sqrt{2 m R_1 r \mu}  \twonorm{y - y'}.
\end{align*}
Thus,
\begin{align*}
  N(B \times D, d, \epsilon) &\leq N\left(B, \twonorm{\cdot}, \frac{\epsilon}{2\sqrt{2 m R_2 s \mu}}\right)\cdot N\left(D, \twonorm{\cdot}, \frac{\epsilon}{2\sqrt{2 m R_1 r \mu}}\right)\\
  &\leq  \left(\frac{6 \sqrt{2 m R_2 s
        \mu}}{\epsilon}\right)^{s} \cdot {n \choose r} \left(\frac{6 \sqrt{2m R_1 r
        \mu}}{\epsilon}\right)^r.
\end{align*}
The second line comes from the standard volumetric estimate $N(B,
\twonorm{\cdot}, \epsilon) \leq
\left(\frac{3}{\epsilon}\right)^{s}$ for $\epsilon \leq 1$.  The
factor ${n \choose r}$ arises from decomposing $D$ as the union of
${n - s \choose r}$ sets of the same form as $B$, but with support size bounded by $r$. Now, in order to bound the last inequality, we further write ${n
  \choose r} \leq n^r$ and $R_1, R_2 \leq m$.  Plugging this in, we
obtain 
\[
\sqrt{\log N(B \times D, d, \epsilon)} \leq C \sqrt{r + s}
\sqrt{\log(mn/\epsilon)}.
\]
To conclude, a simple integration gives
\[
\int_0^1 \sqrt{r + s} \sqrt{\log(mn/\epsilon)} d\epsilon \leq
\sqrt{r + s} (\sqrt{\log(mn)} + 1), 
\] 
which establishes the claim since the right-hand side is dominated by $\sigma^{-1}$.

\subsection{Coarse scale: $k \le 0$}
\label{sec:coarseScale}

This section contains the crucial estimates, which must be developed
very carefully.  To show that $\varphi_k$ is bounded, observe that by
definition, $\rho^{-k_0 - 1} \leq m^{3/2}$, and thus $-(k_0 + 1) \leq
\log m$ provided that $\log \rho > 3/2$.  It follows that $\varphi_k \leq
1 + (\log m + 1)/\log m \leq 3$.  

Next, we show that the more subtle bound \eqref{eq:majreq} holds.  Let
$\{(x_i, y_i)\}$ be the points in the definition of the Majorizing
measure theorem with mutual distances at least $\rho^{-k-1}$, so that
$\tilde{N} = \abs{\{(x_i, y_i)\}}$.  Let $z_x$ be the minimizer of
$\{\twonorm{z}^2 : \opnorm{z - x}_B \leq \rho^{-k}\}$ and let $z_j$ be
the minimizer of $\{\twonorm{z}^2 : \opnorm{z - x_j}_B \leq \rho^{-k}\}$.
Finally, introduce the pivotal quantity
\[
\theta \coloneq \max_{1 \leq j \leq \tilde{N}} \twonorm{z_j}^2 -
\twonorm{z_x}^2.
\]
We must show that
\[
\rho^{-k} \sigma \sqrt{\log \tilde{N}} \leq \max_{1 \leq j \leq
  \tilde{N}} \varphi_{k+2}(x_j, y_j) - \varphi_k(x,y) 
= \theta + 2/\log m.
\] 
In order to bound $\tilde{N}$, we consider the points $\{z_j, y_j\}$
and note that $\tilde{N} = \abs{\{z_j, y_j\}}$.

We shall need two key properties of the points $\{z_j, y_j\}$.  First,
these points are well separated.  Indeed, the triangle inequality,
gives for $i \neq j$
\begin{align*}
  d((z_i, y_i), (z_j, y_j)) &\geq d((x_i, y_i), (x_j, y_j)) - d((x_i, y_i), (z_i, y_i)) - d((x_j, y_j), (z_j, y_j))\\
  &\geq \rho^{-k-1} - \sqrt{2} \opnorm{x_i - z_i}_B - \sqrt{2}
  \opnorm{x_j - z_j}_B \\
  & \geq \rho^{-k-1} - 2\sqrt{2} \rho^{-k-2}\\ 
& \geq \frac{1}{2} \rho^{-k - 1}
\end{align*}
provided that $\rho \ge 4 \sqrt{2}$. Second, each $z_j$ is close to $x$
in the sense that
\[
\opnorm{x - z_j}_B \leq \opnorm{x - x_j}_B + \opnorm{x_j - z_j}_B \leq
\frac{1}{\sqrt{2}} d((x,y), (x_j, y_j))+ \opnorm{x_j - z_j}_B \leq
\frac{1}{\sqrt{2}} \rho^{-k} + \rho^{-k - 2} \leq \rho^{-k} 
\] provided that $\rho^2 \ge 2 + \sqrt{2}$.  Therefore, it follows from
the definition of $z_x$ that $\twonorm{z_j} \geq \twonorm{z_x}$.

Now, the benefit of the special construction of $\varphi_k$ on the
coarse scale is that the size of $\theta$ restricts the space that
$\{z_j\}$ can inhabit.  To demonstrate this, since $\{z : \opnorm{z -
  x}_B \leq \rho^{-k}\}$ is convex, $\frac{z_x + z_j}{2} \le \rho^{-k}$
belongs to this set.  Now combine $\twonorm{\frac{z_x + z_j}{2}} \geq
\twonorm{z_x}$ with $\twonorm{z_j} \geq \twonorm{z_x}$ to give
\[
\twonorm{\frac{z_j - z_x}{2}}^2 = \frac{1}{2} \twonorm{z_j}^2 +
\frac{1}{2} \twonorm{z_x}^2 - \twonorm{\frac{z_j + z_x}{2}}^2 \leq
\twonorm{z_j}^2 - \twonorm{z_x}^2 \leq \theta.
\] 
Hence,
\[
\twonorm{z_j - z_x} \leq 2 \sqrt{\theta}.
\]
Combined with Lemma \ref{lem:packing}, we obtain
\begin{equation}
\tilde{N} \leq N(B(u, 2\sqrt{\theta}) \times D, d, \rho^{-k - 1}/4),
\end{equation}
where $B(u, 2\sqrt{\theta}) \coloneq \{x : \supp{x} \subset T, \twonorm{x} \leq 1, \twonorm{x - u} \leq 2 \sqrt{\theta}\}$.  

Set $\epsilon = \rho^{-k -1}/4$.  We cover $B(u, 2\sqrt{\theta}) \times
D$ to precision $\epsilon$ in the metric $d$ by covering $B(u,
2\sqrt{\theta})$ to precision $\epsilon/2$ under the norm $\sqrt{2}
\opnorm{\cdot}_B$ and $D$ to precision $\epsilon/2$ under the norm
$\sqrt{2 m R_1} \opnorm{\cdot}_D$.  We have
\begin{equation}
\label{eq:boundNtilde}
\sqrt{\log \tilde{N}} \leq \sqrt{\log N(B(u, 2\sqrt{\theta}), \sqrt{2} \opnorm{\cdot}_B, \epsilon/2)} + \sqrt{\log N(D, \sqrt{2 m R_1} \opnorm{\cdot}_D, \epsilon/2)}.
\end{equation}
To bound the second term, observe that $D \subset \sqrt{r}
B_{\ell_1}$, where $B_{\ell_1}$ is the unit ball under the $\ell_1$
norm.  Hence, 
\[
N(D, \sqrt{2 m R_1} \opnorm{\cdot}_D, \epsilon/2) \leq
N(B_{\ell_1}, \opnorm{\cdot}_D, C \epsilon/\sqrt{r R_1 m})
\]  
and Lemma 3.7 in \cite{rudelsonVershynin} bounds the right-hand side
by following an argument from \cite{carl85}. This lemma
gives\footnote{We do not reproduce the proof here, but encourage
  interested parties to read the clever and short argument.}
\begin{equation}
\label{eq:D_Scover}
\sqrt{\log(N(D, \sqrt{2 m R_1} \opnorm{\cdot}_D, \epsilon/2))} \leq C \sqrt{m} \, \frac{\sqrt{\mu r R_1 \log n \log m}}{\epsilon}.
\end{equation}

Now we bound $N(B(u, 2\sqrt{\theta}), \sqrt{2} \opnorm{\cdot}_B,
\epsilon/2) = N(B, \opnorm{\cdot}_B, \epsilon/4\sqrt{2\theta})$ as
follows:
\begin{equation}
\label{eq:coverB}
\sqrt{\log N(B, \opnorm{\cdot}_B, \epsilon/4\sqrt{2\theta})} \leq C \frac{\sqrt{\theta}}{\epsilon} \sqrt{m s \mu (1 + R_2)}.
\end{equation}
We postpone the proof and show how it implies \eqref{eq:majreq}.  With
$\epsilon = r^{-k}/4$, \eqref{eq:boundNtilde} together with the bounds
\eqref{eq:coverB} and \eqref{eq:D_Scover} give
\[
\sqrt{\log \tilde{N}} \leq C \rho^k \sqrt{m} \Bigl(\log n \sqrt{ \mu r
  R_1} + \sqrt{\mu s (1 + R_2)\theta}\Bigr).
\] 
Now plug in $2\sqrt{\theta} \leq \theta \sqrt{\log m} + 1/\sqrt{\log
  m}$, along with the definition of $\sigma$, to give
\[\rho^{-k} \sigma \sqrt{\log \tilde{N}} \leq \frac{2}{\log m} + \theta\]
as desired, thus concluding the proof of Theorem 6.2.

\begin{proof}[\eqref{eq:coverB}]
  Using the dual Sudakov minorization (Lemma \ref{lem:sudakov}), we
  have
\begin{equation}
\label{eq:sudakovBound}
\sqrt{\log N(B, \opnorm{\cdot}_B,  \epsilon/4\sqrt{2\theta})} 
\leq C \frac{\sqrt{\theta}}{\epsilon} 
\sqrt{\E \sup_{y \in D} \sum_{i=1}^m \<a_i, z_T\>^2 \<a_i, y\>^2}. 
\end{equation}
Since $\<a_{i, T}, z_T\> \sim \mathcal{N}(0, \twonorm{a_{i, T}}^2)$,
$\E \<a_{i, T}, z_T\>^2 = \twonorm{a_{i,T}}^2$.  We write 
\begin{align*}
  \E \sup_{y \in D} \sum_{i=1}^m \<a_{i,T}, z_T\>^2 \<a_i, y\>^2 &= \E \sup_R \opnorm{\sum_{i=1}^m \<a_{i,T}, z_T\>^2 a_{i,R}\, a_{i,R}^*}\\
  &\leq \E \sup_R \opnorm{\sum_{i=1}^m (\<a_{i,T}, z_T\>^2 - \twonorm{a_{i,T}}^2) a_{i,R}\, a_{i,R}^*} + \sup_R \opnorm{\sum_{i=1}^m  \twonorm{a_{i,T}}^2 a_{i,R}\, a_{i,R}^*} \\
  & \coloneq  I_0 + I_1. 
\end{align*}
The supremum is over $R$ obeying $\abs{R} \leq s$ and $R \cap T =
\emptyset$. Since $\twonorm{a_{i,T}}^2 \leq s \mu$, we have
\begin{equation}
\label{eq:boundII}
I_1 \leq s \mu \sup_R \opnorm{\sum_{i=1}^m a_{i,R}\, a_{i,R}^*} 
= s \mu m R_2. 
\end{equation}
Further, since $\E \<a_{i,T}, z_T\>^2 - \twonorm{a_{i,T}}^2 = 0$, we
can use symmetrization as before to obtain
\[
I_0 \leq 2 \E \sup_R \opnorm{\sum_{i=1}^m \xi_i \<a_{i,T}, z_T\>^2
  a_{i,R}\, a_{i,R}^*},
\] 
where $\{\xi_i\}$ is a Rademacher sequence.  We use a lemma --- which is
a direct consequence of Lemma 3.6 in \cite{rudelsonVershynin} --- to
bound this quantity.
 
\begin{lemma}[\cite{rudelsonVershynin}] Let $v_1, \hdots, v_m$, $m
  \leq n$, be vectors in $\mathbb{R}^n$ obeying $\infnorm{v_i}^2 \leq
  \mu_0$ for each $i$.  Let $\{\xi_i\}$ be a Rademacher sequence. Then
\[
\E \sup_{\abs{R} \leq r} \opnorm{\sum_{i=1}^m \xi_i \, v_{i,R}
  \,v_{i,R}^*} \leq k \sup_{\abs{R} \leq r} \opnorm{\sum_{i=1}^m
  v_{i,R}\, v_{i,R}^*}^{\frac{1}{2}}
\] where $k = C \sqrt{r} \log r \sqrt{\log n} \sqrt{\log m} \leq C
\sqrt{r \mu_0 \log n \log^3 m}$.
\end{lemma}

In order to use the lemma, condition on the value of $z_T$ and set
$\mu_0 \coloneq \mu \opnorm{z_T}^2_x$.  Set $G(z) \coloneq \sup_{y \in
  D} \sum_{i=1}^m \<a_i, z_T\>^2 \<a_i, y\>^2$. Absorbing a factor of
2 into $k$, the lemma gives
\begin{equation}
\label{eq:boundI}
I_0 \leq \E k \sqrt{G(z)} \leq \E \frac{k^2}{2} + \E \frac{G(z)}{2}.
\end{equation}

It remains to bound $\E k^2 \leq C r \mu \log n \log^3 m \E
\opnorm{z_T}_D^2$.  Recall that
\[
\opnorm{z_T}_D^2 \coloneq \max_{1\leq i \leq m} \abs{\<a_{i,T},
  z_{T}\>}^2
\]
and set $\bar{\sigma}^2 \coloneq \max_i \twonorm{a_{i,T}}^2 \leq
s \mu$.  It now follows from a standard concentration bound on
subexponential random variables that
\[
\E \opnorm{z_T}_D^2 \leq C \bar{\sigma}^2 \log m \leq C s \mu
\log m
\] 
(this can be derived by bounding $\P(\opnorm{z_T}_D^2 > t)$ for all
$t$ and integrating).

Finally, plug this last bound into \eqref{eq:boundI}, and combine the
result with \eqref{eq:boundII}. This gives 
\[
\E G(z) \leq C r s  \mu^2 \log n \log^4 m + s \mu m R_2 + \E
G(z)/2.
\]
Rearranging the terms together with $m \geq C r \mu \log n \log^4 m$ give
\[
\E G(z) \leq C m s \mu (1 + R_2).
\] 
Finally, inserting this into \eqref{eq:sudakovBound} gives the conclusion. 
\end{proof}

\subsection{Concentration around the mean}

We have now proved that $\E X \leq \epsilon$ for any $\epsilon > 0$
provided that $m \geq C_\epsilon \, \mu \, [s \log m \vee r
\log n \log^4 m]$. This already shows that for any fixed $\delta > 0$,
\[
\P(X > \delta) \leq \frac{\epsilon}{\delta}
\]
and so taking $\epsilon$ to be a small fraction of $\delta$ gives a
first crude bound. However, we wish to show that if $m \geq C
\mu \, \beta \, [s \log m \vee r \log n \log^4 m]$ then the probability of
`failure' decreases as $e^{-\beta}$.  This can be proved using a
theorem of \cite{rudelsonVershynin} which in turn is a combination of
Theorem 6.17 and inequality (6.19) of \cite{ledouxTalagrand}.  We
restate this theorem below.
\begin{theorem}
  Let $Y_1, \hdots, Y_m$ be independent symmetric random variables
  taking values in some Banach space.  Assume that $\opnorm{Y_j} \leq
  R$ for all $j$ and for some norm $\opnorm{\cdot}$.  Then for any
  integer $\ell \geq q$, and any $t > 0$, the random variable
\[
Z \coloneq \opnorm{\sum_{j=1}^m Y_j}
\]
obeys
\[
\P(Z \geq 8 q \E Z + 2 R \ell + t) \leq \left(\frac{C}{q}\right)^\ell + 2
\exp\left(- \frac{t^2}{256 q (\E Z)^2}\right).
\]
\end{theorem}

In our setup, we work with a norm on positive semidefinite matrices
given by
\[
\opnorm{M} \coloneq \sup_{v \in V}\,\, v^* M v, 
\]
where $V$ is given by \eqref{eq:V}.  The rest of the details of the
proof of concentration around the mean follows exactly as in the steps
of \cite[pages 11-12]{rudelsonVershynin} and so we do not repeat them,
but encourage the interested reader to check \cite{rudelsonVershynin}.
This is the final step in proving Theorem \ref{teo:weakRIP}.

\section{Stochastic Incoherence}

In Sections 2--4, we have assumed that the coherence bound holds
deterministically, and it is now time to prove our more general
statement; that is to say, we need to extend the proof to the case
where it holds stochastically. We propose a simple strategy: condition
on the (likely) event that each row has `small' entries, as to
recreate the case of deterministic coherence (on this event).  Outside
of this event, we give no guarantees, but this is of little
consequence because we will require the event to hold with probability
at least $1 - 1/n$.  A difficulty arises because the conditional
distribution of the rows no longer obeys the isotropy condition
(although the rows are still independent).  Fortunately, this
conditional distribution obeys a \textit{near isotropy condition}, and
all of our results can be reproved using this condition instead. In
particular, all of our theorems follow (with adjustments to the
absolute constants involved) from the following two conditions on the
distribution of the rows:
\begin{equation}
\begin{array}{rll}
  \opnorm{\E aa^* - \Id} & \leq 1/(8\sqrt{n})  & \qquad \text{(near isotropy)}\\
  \max_{1 \le t \le n}\twonorm{a[t]}^2 & \leq \mu  & \qquad \text{(deterministic coherence)}. 
  \end{array}
\end{equation}
We first illustrate how to use near isotropy to prove our results.
There are several results that need to be reproved, but they are all
adjusted using the same principle, so to save space we just prove that
a slight variation on Lemma \ref{lem:localConditioning} still holds
when requiring near isotropy, and leave the rest of the analogous
calculations to the interested reader.

Set $W \coloneq \E aa^*$ and let $W_{T,T}$ be the restriction of $W$
to rows and columns in $T$.  We first show that
\begin{equation}
\label{eq:Wbound}
\P(\opnorm{A^*_T A_T - W_{T, T}} \geq \delta) \leq 2 s \exp\left(-\frac{m}{\mu (s+1)} \, \frac{\delta^2}{2 +  2\delta/3}\right). 
\end{equation}
To prove this bound, we use the matrix Bernstein inequality of Section
\ref{sec:estimate1}, and also follow the framework of the calculations
of Section \ref{sec:estimate1}.  Thus, we skim the steps.  To begin,
decompose $A^*_T A_T - W_{T,T}$ as follows:
\[
m(A^*_T A_T - W_{T,T}) = \sum_{k=1}^m (a_{k,T} a_{k,T}^* - W_{T,T})
\coloneq \sum_{k=1}^m X_k.
\]
We have $\E X_k = 0$ and $\opnorm{X_k} \leq \opnorm{a_{k, T} a_{k,
    T}^* - I} + \opnorm{I - W_{T,T}} \leq s \mu + \frac{1}{8 \sqrt{n}}
\leq (s+1) \mu \coloneq B$.  Also, the total variance obeys
\[\opnorm{\E X_k}^2 \leq \opnorm{\E (a_{k, T} a_{k, T}^*)^2} \leq s \mu \opnorm{\E a_{k, T} a_{k, T}^*} = s \mu \opnorm{W_{T,T}} \leq s \mu (1 + \tfrac{1}{8 \sqrt{n}}) \leq (s+1) \mu.\]
Thus, $\sigma^2 \leq m (s+1) \mu$, and \eqref{eq:Wbound} follows from
the matrix Bernstein inequality.

Now, it follows from $\opnorm{W_{T,T} - \Id} \leq \opnorm{W - \Id}
\leq \frac{1}{8 \sqrt{n}}$ that
\[
\P\Bigl(\opnorm{A^*_T A_T - \Id} \geq \tfrac{1}{8 \sqrt{n}} +
\delta\Bigr) \leq 2 s \exp\left(-\frac{m}{\mu (s+1)} \,
  \frac{\delta^2}{2 + 2\delta/3}\right).
\]
In the course of the proofs of Theorems \ref{teo:noiseless} and
\ref{teo:noisy} we require $\opnorm{A^*_T A_T - \Id} \leq 1/2$ for
noiseless results and $\opnorm{A^*_T A_T - \Id} \leq 1/4$ for noisy
results. This can be achieved under the near isotropy condition by
increasing the required number of measurements by a tiny bit.  In
fact, when proving the analogous version of Lemma
\ref{lem:localConditioning}, one could weaken the near isotropy
condition and instead require $\opnorm{\E aa^* - \Id} \leq 1/8$, for
example.  However, in extending some of the other calculations to work
with the near isometry condition --- such as \eqref{eq:p_2} --- the factor
of $\sqrt{n}$ (or at least $\sqrt{s}$) in the denominator appears
necessary; this seems to be an artifact of the method of proof,
namely, the golfing scheme.  It is our conjecture that all of our
results could be established with the weaker requirement $\opnorm{\E
  aa^* - \Id} \leq \epsilon$ for some fixed positive constant
$\epsilon$.

We now describe the details concerning the conditioning on rows having
small entries.  Fix the coherence bound $\mu$ and let
\[E_k = \left\{\max_{1 \leq t \leq n} \abs{a_k[t]}^2 \leq \mu\right\}
\qquad \text{and} \qquad G = \cap_{1 \le k \le m} \,\, E_k.\] Thus $G$
is the `good' event ($G$ is for good) on which $\max_{1 \leq t \leq n}
\abs{a_k[t]}^2 \le \mu$ for all $k$.  By the union bound, $\P(G^c)
\leq m \P(E_1^c)$.  We wish for $\P(G^c)$ to be bounded by $1/n$, and
so we require $\mu$ to be large enough so that $\P(E_1^c) \leq (m
n)^{-1}$.

Next we describe how conditioning on the event $G$ induces the near
isometry condition.  Because of the independence of the rows of $A$,
we may just consider the conditional distribution of $a_1$ given
$E_1$.  Drop the subindex for simplicity and write
\[\Id = \E[a a^*] = \E[a a^* \mathbb{1}_E] + \E[aa^* \mathbb{1}_{E^c}] = \E[a a^*|E] \P(E) + \E[aa^* \mathbb{1}_{E^c}].\]
Thus, 
\begin{equation}
\label{eq:earlyCoherence}
\opnorm{\E[aa^*|E] - \Id} \cdot \P(E) = \opnorm{(1 - \P(E)) \Id - \E[aa^* \mathbb{1}_{E^c}]} \leq \P(E^c) + \opnorm{\E[aa^* \mathbb{1}_{E^c}]}.
\end{equation}
We now bound $\opnorm{\E[aa^* \mathbb{1}_{E^c}]}$. By Jensen's
inequality (which is a crude, but still fruitful, bound here),
\begin{equation}
\label{eq:mediumCoherence}
\opnorm{\E[aa^* \mathbb{1}_{E^c}]} \leq \E[ \opnorm{aa^* \mathbb{1}_{E^c}}] = \E[ \twonorm{a}^2 \mathbb{1}_{E^c}]. 
\end{equation}
and, therefore,
\[
\opnorm{\E[aa^*|E] - \Id} \leq \frac{1}{1 - \P(E^c)} \left( \P(E^c) +
  \E[\twonorm{a}^2 \mathbb{1}_{E^c}]\right).
\]
Combine this with the requirement that $\P(E^c) \leq (m n)^{-1}$ to
give
\[\opnorm{\E[aa^*|E] - \Id} \leq \frac{19}{20} \left(\frac{1}{20 \sqrt{n}} + \E[ \twonorm{a}^2 \mathbb{1}_{E^c}]\right)\]
as long as $m \sqrt{n} \geq 20$.  It now follows that in order to ensure near isotropy, it is sufficient that
\[\E[ \twonorm{a}^2 \mathbb{1}_{E^c}] \leq \frac{1}{20 \sqrt{n}}.\]
It may be helpful to note a simple way to bound the left-hand side
above.  If $f(t)$ is such that
\[
\P\left(\max_{1 \leq t \leq n} \abs{a[t]}^2 \geq t\right) \leq f(t),
\]
then a straightforward calculation shows that
\[
\E[ \twonorm{a}^2 \mathbb{1}_{E^c}] \leq n \mu f(\mu) + n
\int_{\mu}^\infty f(t) dt.
\]

\small

\subsection*{Acknowledgements}
This work has been partially supported by ONR grants N00014-09-1-0469
and N00014-08-1-0749, and by the 2006 Waterman Award from NSF.  We
would like to thank Deanna Needell for a careful reading of the
manuscript. 

\bibliographystyle{plain}
\bibliography{myref}

\end{document}